\documentclass{article}
\usepackage[utf8]{inputenc}
\usepackage{amssymb}
\usepackage[T1]{fontenc}
\usepackage{amsmath}
\usepackage{amsthm}
\usepackage{xcolor}
\usepackage{graphicx}

\usepackage[caption=false]{subfig} 

\usepackage{enumerate}
\usepackage[shortlabels]{enumitem}
\usepackage{hyperref} 
\newtheorem{definition}{Definition}[section]
\newtheorem{remark}{Remark}[section]
\newtheorem{lemma}{Lemma}
\newtheorem{theorem}[lemma]{Theorem}
\newtheorem*{theoremA}{Theorem A}
\newtheorem*{theoremB}{Theorem B}
\newtheorem{proposition}[lemma]{Proposition}

\def\R{\mathbb{R}}

\DeclareMathOperator{\tr}{tr}

\title{The impact of high frequency-based stability on the onset of action potentials in neuron models\footnotemark[1]}
\author{Eduardo Cerpa\footnotemark[3] , Nathaly Corrales\footnotemark[7], \,Mat\'ias Courdurier\footnotemark[4],\\
\, Leonel E. Medina\footnotemark[6], \, Esteban Paduro\footnotemark[3]\,\,\footnotemark[2]}
\date{February 2024}

\begin{document}
\maketitle

\footnotetext[1]{This work has been partially supported by ANID Millennium Science Initiative
Program through Millennium Nucleus for Applied Control and Inverse Problems NCN19-161. 
}
\footnotetext[2]{Corresponding author: Esteban Paduro. Email: {\tt esteban.paduro@mat.uc.cl}}

\footnotetext[3]{Instituto de Ingenier\'ia Matem\'atica y Computacional, Facultad de Matem\'aticas, Pontificia Universidad Cat\'olica de Chile,  Avda. Vicu\~na Mackenna 4860, Santiago, Chile.
}
\footnotetext[4]{Departamento de Matem\'atica, Facultad de Matem\'aticas, Pontificia Universidad Cat\'olica de Chile. Avda. Vicu\~na Mackenna 4860, Santiago, Chile.
}
\footnotetext[6]{Departamento de Ingenier\'ia Inform\'atica, Universidad de Santiago de Chile, Avda. V\'ictor Jara 3659, Santiago, Chile. 
}
\footnotetext[7]{Departamento de Matem\'aticas, Universidad T\'ecnica Federico Santa Mar\'ia, Avda. Vicu\~na Mackenna 3939, Santiago, Chile.
}

\begin{abstract}
This paper studies the phenomenon of conduction block in model neurons using high-frequency biphasic stimulation (HFBS). The focus is investigating the triggering of undesired onset action potentials when the HFBS is turned on. The approach analyzes the transient behavior of an averaged system corresponding to the FitzHugh-Nagumo neuron model using Lyapunov and quasi-static methods. The first result provides a more comprehensive understanding of the onset activation through a mathematical proof of how to avoid it using a ramp in the amplitude of the oscillatory source. The second result tests the response of the blocked system to a piecewise linear stimulus, providing a quantitative description of how the HFBS strength translates into conduction block robustness. The results of this work can provide insights for the design of electrical neurostimulation therapies.
\end{abstract}

\vspace{0.2 cm}


\vspace{0.2 cm} {\bf AMS subject classifications 2020:} 37N25, 92-10.

\section{Introduction}
In neuroscience, an action potential is a rapid rise and fall of the membrane voltage that travels along a neuron. Action potentials are considered the fundamental unit of communication of the nervous system, but in some cases, their presence is undesirable, \textit{e.g.}, abnormal activity resulting in pain \cite{neudorferKilohertzfrequencyStimulationNervous2021a}. Conduction block is a phenomenon in which action potentials are prevented from traveling along a nerve fiber, thereby canceling such unwanted activity. This can be achieved via the application of a direct current (DC) signal \cite{bhadraDirectCurrentElectrical2004} or high-frequency biphasic stimulation (HFBS), typically in the kilohertz range \cite{tanner_reversible_1962}.  A better understanding of the conduction block phenomenon is not only relevant from the theoretical point of view but also for medical applications \cite{taiSimulationAnalysisConduction2005a, kilgore_nerve_2004, Bhadra2007}. In particular, a collateral effect of this technique that is not completely understood is a finite burst of action potentials that appear immediately after the HFBS is turned on, called ``onset response''. In this paper, we advance the understanding of the conduction block phenomenon based on its relation with notions of stability for a corresponding averaged system. This interpretation adds to the existing biophysical explanations known in the literature \cite{ackermannDynamicsSensitivityAnalysis2011}. 

The conduction block phenomenon has been studied using varied neuron models \cite{Weinberg2013, Ratas2012,taiSimulationAnalysisConduction2005a,milesEffectsRampedAmplitude2007a}. In this study, we seek explicit calculations; hence, we choose a tractable model known as the FitzHugh-Nagumo (FHN) system \cite{FitzHugh1961}
\begin{equation}\label{FHN_hf_intro}
\left\{
\begin{array}{rl}
\dot{v}&= v -v^3/3 - w +I(t),\\
\dot{w}&= \varepsilon(v- \gamma w + \beta),\\
v(0) &= v_0,\quad w(0) = w_0,
\end{array}\right.
\end{equation}
with an input current 
\begin{equation}\label{current_no_slope}
I(t) = I_0 + \rho \omega \cos(\omega t),
\end{equation}
where $v$ is the membrane voltage and $w$ is a recovery variable. 

From a mathematical perspective, it is natural to study HFBS conduction block for the FHN system via the method of averaging \cite{khalil2013}.  Recently, it was established that there is a threshold for the DC component, $I_0$, that depends on the amplitude $\rho$ of the oscillatory input, above which persistent excitation is observed, whereas, for sub-threshold stimuli, only a finite number of action potentials appear  \cite{Weinberg2013}. A limitation of this analysis is that only long-term behavior can be predicted by looking at the stability of critical points for the averaged system. This work extends these ideas to analyze the transient behavior of the system, \textit{i.e.}, when the onset response is observed. In medical applications, the onset response to HFBS  is undesirable because it may cause pain or discomfort \cite{neudorferKilohertzfrequencyStimulationNervous2021a}. To avoid the onset response, several authors have considered the use of a gradual increase of the HFBS amplitude \cite{taiSimulationAnalysisConduction2005a, milesEffectsRampedAmplitude2007a, vrabecReductionOnsetResponse2019, gergesFrequencyAmplitudetransitionedWaveforms2010, ackermannConductionBlockWhole2010,vrabecNovelWaveformNoOnset2013}. To understand the potential benefits of such a strategy, we study the transient response of system \eqref{FHN_hf_intro} to ramped HFBS. We consider waveforms as illustrated in Figure \ref{figure_input_current}:
A ramped HFBS is used to avoid action potentials, and a ramped DC signal is used to test the conduction block of the system. The main results of this work can be roughly stated as follows.

\begin{theoremA}[Onset effects can be avoided using ramped HFBS]
Under suitable stability conditions, it is possible to avoid the onset response by slowly increasing the HFBS amplitude.
\end{theoremA}

\begin{theoremB}[Transient response of the blocked system to DC stimulation]
Consider an FHN neuron where HFBS has been applied for a sufficiently long time. Suppose a DC stimulus is applied to this FHN neuron and does not generate \textit{persistent} action potentials because the HFBS is strong enough. Then, if the DC term is applied with its amplitude increasing as a ramp, such a ramp has a maximum slope that results in no action potential in response to the stimulus.
\end{theoremB}

The precise statements and hypotheses in these results will be given in Theorem \ref{thm_slow_steering1} and Theorem \ref{thm_slow_steering2}, respectively. The importance of Theorem A is that it supports existing experimental results and provides additional tools to understand this phenomenon. Theorem B deals with the robustness of the HFBS conduction block by exploring how fast we can change the DC term and still avoid action potentials.

\subsection{Main Results}

We proceed to a more precise description of our results.
As will be detailed later, we can recast our problem as two steering problems:

\begin{figure}
\centering
\includegraphics[width=0.8\linewidth]{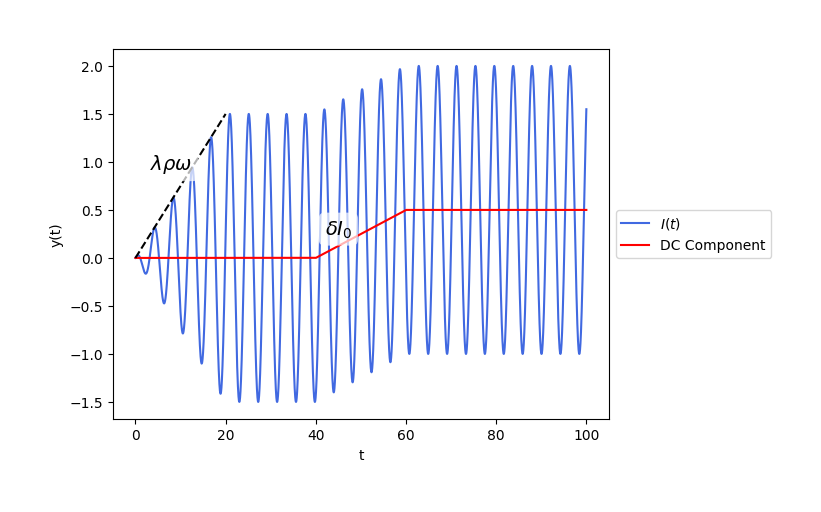}
\caption{Example of input current described in \eqref{input_current_slope_general1} and \eqref{input_current_slope_general2} during $T= 100$. In this case,  $\omega = 1.5$, $\rho = 1$, $I_0 = 0.5$,  $\lambda = 0.05$, $\delta = 0.05$. Current $I_1(t)$ is applied until $t =$ 40, followed by the application of current $I_2(t)$ until $t =$ 100.}\label{figure_input_current}
\end{figure}

\begin{itemize}
\item {\bf Steering Problem 1:} Find a value of $\lambda >0$ so that by applying a current
\begin{equation}\label{input_current_slope_general1}
I_1(t) = 
\begin{cases}
t \lambda \omega \rho \cos(\omega t) &,\quad  0\leq t \leq 1 /\lambda,\\
 \omega \rho \cos(\omega t) &,\quad  t >1/\lambda ,
\end{cases}
\end{equation}
we can steer the solutions of system \eqref{FHN_hf_intro} from a neighborhood of $P_0=(v_0,w_0)$ given by
\begin{equation}\label{definition_P0}
\left\{
\begin{array}{rl}
0&= v_0- v_0^3/3-w_0,\\
0&= v_0 - \gamma w_0 + \beta,
\end{array}
\right.
\end{equation}
to an oscillating trajectory centered at the point $P_1=(v_1,w_1)$ given by
\begin{equation}\label{definition_P1}
\left\{
\begin{array}{rl}
0&= (1-\rho^2/2)v_1- v_1^3/3-w_1, \\
0&= v_1- \gamma w_1 + \beta,
\end{array}
\right.
\end{equation}
without generating action potentials. 

\item {\bf Steering problem 2:} Find a value of $\delta >0$ so that by applying a current
\begin{equation}\label{input_current_slope_general2}
I_2(t) = 
\begin{cases}
\delta t I_0 + \rho \omega \cos(\omega t)&,\quad  t \leq 1/\delta,\\
I_0 + \rho \omega \cos(\omega t)&, \quad t >  1/\delta,
\end{cases}
\end{equation}
we can steer the solutions of system \eqref{FHN_hf_intro} from an oscillating trajectory centered at $P_1$ given by \eqref{definition_P1} to an oscillating trajectory centered at  $P_2=(v_2,w_2)$ given by
\begin{equation}\label{definition_P2}
\left\{
\begin{array}{rl}
0&= (1-\rho^2/2)v_2- v_2^3/3-w_2 + I_0,\\
0&= v_2 - \gamma w_2 + \beta,
\end{array}
\right.
\end{equation}
without generating action potentials.

\end{itemize}

Our work presents two main contributions. First, we propose a new explanation for the onset response in the conduction block phenomenon and how to avoid it using ramped HFBS.
Second, we give a better understanding of the transient effects of DC stimuli when the input currents \eqref{input_current_slope_general1} and \eqref{input_current_slope_general2} are considered.

Our results are based on averaging techniques to separate the problem into fast and slow timescales. 
The averaging method used here builds upon our recent work \cite{cerpa_partially_2022}, albeit more subtle since a global-in-time argument is needed. In the slow timescale, the partially averaged system gives a non-autonomous system, 
which is critical to understanding the effects of the DC component of the input signal in steering problem 2. Further, to study the steering part of the problem, we argue that for sufficiently small slopes, the problem can be analyzed using quasi-static steering arguments \cite{MR2086173, khalil2013}, which require some assumptions on the parameters so that the time-frozen systems satisfy some stability conditions. 

The method of quasi-static steering indicates that whenever the ramp parameters $\lambda$ and  $\delta$ are suitably small, we can compare the solutions of system \eqref{FHN_hf_intro} with the solution of a more straightforward algebraic system. Since we work with functions that include a ramp in one of the factors, it is helpful to introduce the following notation.

\begin{definition}
We denote by $S(t)$ the ramp-type function given by 
\begin{equation}\label{cutoff}
S(t) = \begin{cases}
0 &, t<0,\\
t & , 0\leq t \leq 1,\\
1 &, t >1.
\end{cases}
\end{equation}
\end{definition}

\begin{definition}[Approximate solution]\label{defi_algebraic_eqn}
Suppose the parameters $\beta$, $\gamma$, $\rho$, $I_0$ are adequate (such that the  systems \eqref{PAS_slope_freq_instantaneous_alpha} and \eqref{instantaneous_system_alpha} have a unique solution). We define  the real-valued functions
\[
X_1,Y_1,X_2,Y_2: [0,\infty) \to \R ,
\]
given by the solution of the following algebraic equations.
\begin{itemize}
\item For steering problem 1, given $\alpha\geq 0$, the reference trajectory is given by
\begin{equation}\label{PAS_slope_freq_instantaneous_alpha}
\left\{
\begin{array}{rl}
0&= (1 - S(\alpha)^2 \rho^2 /2 )\tilde{X}_1(\alpha)-\tilde{X}_1(\alpha)^3/3  -\tilde{Y}_1(\alpha), \\
0&= \tilde{X}_1(\alpha)- \gamma \tilde{Y}_1(\alpha) + \beta, 
\end{array}\right.
\end{equation}
and the trajectory scaled via a slope parameter $\lambda >0$ as
\begin{equation}\label{PAS_slope_freq_instantaneous}
X_1(t) =\tilde{X}_1(\lambda t),\quad Y_1(t) = \tilde{Y}_1(\lambda t), \quad t \geq 0.
\end{equation}
Note that $(X_1(0),Y_1(0))=(v_0,w_0)$, $(X_1(1/\lambda) , Y_1(1/\lambda) )= (v_1, w_1)$ as defined by \eqref{definition_P0}, \eqref{definition_P1}.
\item For steering problem 2, given $\alpha \geq0$, the reference trajectory is given by
\begin{equation}\label{instantaneous_system_alpha}
\left\{
\begin{array}{rl}
0&= (1-\rho^2/2)\tilde{X}_2(\alpha) -\tilde{X}_2(\alpha)^3/3 - \tilde{Y}_2(\alpha) + S(\alpha)I_0,\\
0&= \tilde{X}_2(\alpha)- \gamma \tilde{Y}_2(\alpha) + \beta, 
\end{array}\right.
\end{equation}
and the trajectory scaled via a slope parameter $\delta >0$ as
\begin{equation}\label{instantaneous_system}
X_2(t) =\tilde{X}_2(\delta t),\quad Y_2(t) = \tilde{Y}_2(\delta t), \quad t \geq 0.
\end{equation}
Note that $(X_2(0),Y_2(0))=(v_1,w_1)$, $(X_2(1/\delta) , Y_2(1/\delta) )= (v_2, w_2)$ as defined by \eqref{definition_P1}, \eqref{definition_P2}.
\end{itemize}
Additionally, we can guarantee that the functions $X_1(t)$, $Y_1(t)$, $X_2(t)$, and $Y_2(t)$ are well-defined and are continuous because of the inverse function theorem.
\end{definition}

Theorem \ref{thm_slow_steering1} establishes precisely that by taking a slight slope $\lambda >0$ we can always reach the maximum amplitude of the HFBS without generating onset action potentials.

\begin{theorem}[Steering problem 1:  slope in the envelope of the HFBS]\label{thm_slow_steering1}
Let $(\varepsilon, \beta, \gamma, \rho)$  satisfy Condition \ref{conditionC1} and Condition \ref{conditionC3} in Definition \ref{definition_condition_parameters}. Then given $\eta > 0$, there exists $\mu^*,\lambda^* > 0$ such that for all $0<\lambda < \lambda^*$, and all $\omega \geq \omega_0(\lambda)$, for some $\omega_0(\lambda) >0$,  the corresponding solution $(v,w)\in C^1((0,\infty);\R^2)\cap C([0,\infty);\R^2)$ of system \eqref{FHN_hf_intro} with an input current of the form \eqref{input_current_slope_general1} and initial data such that
\[
\|(v(0)-v_0, w(0)-w_0)\|_2 \leq \mu^*,
\]
will satisfy the estimate
\[
\|(v(t) -X_1(t) - \rho~ S(\lambda t) \sin(\omega t), w(t) - Y_1(t)\|_2 \leq \eta,\quad t \geq 0,
\]
where $X_1(t)$ and $Y_1(t)$ are given by Definition \ref{defi_algebraic_eqn}. In other words, for small values of the slope of the envelope of the HFBS, $\lambda>0$, we can guarantee that the solution $(v,w)$ stays sufficiently close to the reference trajectory $(X_1(t) + \rho~ S(\lambda t) \sin(\omega t), Y_1(t))$ and therefore no action potential are generated for any $t \geq 0$. 

\end{theorem}
\begin{remark}
While Conditions \ref{conditionC1} and \ref{conditionC3} are very explicit on what is needed for the result, they are not easy to verify. Proposition \ref{proposition_parameters}  provides a concrete way of ascertaining the conditions. The same is true for the assumptions of Theorem \ref{thm_slow_steering2}.
\end{remark}

\begin{remark}
    While the dependence of the maximum slope in the different parameters of the system is complicated, we will see in Proposition \ref{prop_slope_freq} that we can connect it with quantities related to the system's stability.
\end{remark}

\begin{remark}
Due to the oscillatory nature of the input current, $(v_1, w_1)$ is not an equilibrium point of the system, and therefore, Theorem  \ref{thm_slow_steering1} must be understood as if the system has an attractor near the periodic trajectory
\[
(v_1 + \rho \sin(\omega t), w_1).
\]
\end{remark}

Theorem \ref{thm_slow_steering1} tells us that in addition to having long-time stability (known because of \cite{Weinberg2013}), we can modify the HFBS using a ramp to avoid the onset response.  Moreover, small changes in the proof allow the slope in Theorem \ref{thm_slow_steering1} to be chosen piecewise constant, allowing it to be adjusted dynamically if needed, with the requirement of not being too large (concerning certain stability conditions). 

\begin{theorem}[Steering problem 2: slope in the DC component]\label{thm_slow_steering2}
Let $(\varepsilon$, $\beta$, $\gamma$, $\rho$, $I_0$) satisfy Condition \ref{conditionC2} and Condition \ref{conditionC4} in Definition \ref{definition_condition_parameters}. Then given $\eta > 0$, there exists  $\mu^*, \delta^*, \omega_0 >0$ such that for all $0<\delta < \delta^*$, $\omega \geq \omega_0$
the corresponding solution $(v,w)\in C^1((0,\infty);\R^2) \cap C([0,\infty);\R^2)$ of system \eqref{FHN_hf_intro} with an input current of the form \eqref{input_current_slope_general2}  and initial condition such that
\[
\|(v(0) - v_1, w(0) - w_1)\|_2 \leq  \mu^*,
\]
will satisfy the estimate
\[
\|(v(t) -X_2(t) - \rho \sin(\omega t), w(t) - Y_2(t)\|_2 \leq\eta ,\quad t \geq 0,
\]
where $X_2(t)$ and $Y_2(t)$ are given by Definition \ref{defi_algebraic_eqn}.  In other words, for small values of the slope of the ramp, $\delta >0$, we can guarantee that the solution $(v,w)$, stays sufficiently close to the reference trajectory $(X_2(t) + \rho \sin(\omega t) , Y_2(t))$ and therefore no action potential is generated for any $t \geq 0$.
\end{theorem}

We derive a simpler approximate model via a partial averaging argument to prove the result above. Next, using a quasi-static steering argument on the approximated model, we show that introducing a slow ramp-up on the input current can avoid undesired action potentials and that increasing the amplitude of the HFBS allows increasing the maximum slope of this ramp-up. Finally, we prove that the approximation is valid for all time, and consequently, the original model also does not generate undesired action potentials.

\begin{remark}
   Similar to Theorem \ref{thm_slow_steering1}, it is possible to follow the size of the ramp slope along the proof and see in Proposition \ref{lemma_slow_steer} in a more precise way how it is connected to the stability of the system.
\end{remark}

\begin{remark}\label{remark_consideration_phase}
It is important to notice that there is a consideration about the phase at time $t=0$. If, instead, we consider the input current to be
\[
I(t) = \rho \omega \cos(\omega (t+\phi)) + I_0
\]
such difference of phase is translated directly to the output by instead comparing the solution with 
\[
(X_2(t) + \rho\sin(\omega(t+\phi)) , Y_2(t))
\]
at both the initial time and for $t\geq 0 $. The proof only requires a small modification in Subsection \ref{subsec_approximation_dc_term}. The assumption in the initial condition is not very restrictive because Theorem \ref{thm_slow_steering1} shows that the solution follows the oscillatory part for all time $t\geq 0$.
\end{remark}

A brief comment about well-posedness: in all the problems considered, we always considered ODE where the right-hand side is locally Lipschitz, which guarantees that the solution for the initial value problem is unique in $C([0,T]) \cap C^1((0,T))$ for any $T>0$. Because of this, unless we need to be specific about the interval of validity of the equation, it will be assumed that the solution belongs to $C([0,T]) \cap C^1((0,T))$ for any $T>0$.

\subsection{Known results}
The onset response to HFBS has been identified in {\it in vivo} preparations~\cite{Kilgore2014a}, and different strategies have been proposed to mitigate or eliminate it. Conduction block using HFBS requires higher amplitudes for higher frequencies, and this idea was used to minimize the onset response by initiating the HFBS with higher frequency (30~kHz) and then decreasing it (10~kHz), maintaining conduction block~\cite{gergesFrequencyAmplitudetransitionedWaveforms2010}. Similarly, the amplitude of the HFBS signal has been manipulated so that it started from a non-zero value and then slowly increased in a ramp~\cite{vrabecReductionOnsetResponse2019}. However, it appears that this approach can only effectively eliminate subsequent onset responses once a fiber has been entrained in a partial conduction block state following an initial block state that, nonetheless, exhibited an onset response. Thus, further examination is needed to determine if it is possible to eliminate the first onset response and under what conditions.  In addition, the combination of HFBS on top of a DC component showed great promise in eliminating the onset response~\cite{ackermannConductionBlockWhole2010,vrabecNovelWaveformNoOnset2013}, but there remain challenges related to possible tissue damage due to charge accumulation concomitant with charge unbalanced signals. 

Modeling studies have also provided some information about the onset response. Theoretical findings using a modified HH-based fiber model suggested that the onset response may be efficiently mitigated by adapting the slope of a ramped HFBS signal after monitoring sodium channels in an open loop design~\cite{yi_kilohertz_2020}. Further, increasing the amplitude of the HFBS signal in steps can avoid onset responses in an unmyelinated HH fiber model\cite{zhongHighfrequencyStimulationInduces2022}. This effect was also studied in a morphologically detailed model of a mammalian peripheral nerve fiber~\cite{milesEffectsRampedAmplitude2007a}, and tiny ramps are required to avoid the onset response in this model. 

From the point of view of differential equations, averaging techniques have been used to analyze the FHN system with highly oscillatory sources by exploiting certain stability properties of the averaged system \cite{Weinberg2013,cerpa_partially_2022,Ratas2012,cerpaApproximationStabilityResults2023a}. In particular, for the conduction block problem \cite{Ratas2012,Weinberg2013}, we can discriminate between a finite number of action potentials and persistent action potentials, but such techniques are unsuitable for studying transient behavior. To address this issue, we employ a more refined version of averaging in conjunction with a more general input current  \eqref{input_current_slope_general1} and \eqref{input_current_slope_general2}, which allows us to provide more precise analysis for transient responses using Lyapunov-type arguments \cite{rauch-smoller,Kostova2004, MR2086173}.

There are two key concepts to interpret in our results. First, we know that we can study the conduction block on system \eqref{FHN_hf_intro} using stability properties of the averaged system \cite{Weinberg2013}. The stability notion that we get for time $t > \frac{1}{\lambda}$ is exponential stability since, on that regime, classical averaging theory applies, and we have an exponential equilibrium for the averaged system. Next, we can use Theorems \ref{thm_slow_steering1} and  \ref{thm_slow_steering2} to explain some results for the conduction block phenomenon.

In addition, even though we use the less complex FHN system in this work, we expect the qualitative analysis to be similar when applied to other models. So, our conclusions may be compared with, for example, prior work that involved ramped HFBS signals to HH-based models.

\begin{enumerate}
\item  The transient effects can be completely eliminated using a ramp in a modified HH-type fiber but not in a pure HH fiber \cite{yi_kilohertz_2020}. This suggests that the averaged system of the modified HH is more stable.
\item In a morphologically-detailed HH-based model fiber, standard ramps in HFBS signals are insufficient to eliminate the onset response   \cite{milesEffectsRampedAmplitude2007a}. However, later work suggested that a simply much smaller ramp may eliminate the onset response \cite{zhongHighfrequencyStimulationInduces2022}, which implies that not having enough stability for the averaged system is very restrictive on the use of a ramp.
\item  The onset response may also be avoided via the combined action of HFBS and a DC term~\cite{ackermannConductionBlockWhole2010,vrabecNovelWaveformNoOnset2013}. This can be understood as if the DC component provides additional stability around the equilibrium, enabling a steeper ramp on the oscillatory term after reaching the average system's new equilibrium. At this point, the DC component is no longer needed. 
\end{enumerate}

\subsection{Organization of the paper}

The strategy of our approach can be separated into three steps: 
\begin{itemize}

    \item Step 1. In Section \ref{section_abstract_steering_problem}, we study the abstract steering problem and show that a Lyapunov function for a system with coefficients that vary slowly over time can be constructed under certain conditions related to stability. 
    \item Step 2. In Section \ref{section_quasi_static_steer}, we introduce the partial averaging for the problem under consideration and show that they satisfy the requirements of the results of Section \ref{section_abstract_steering_problem}.
    \item Step 3. In Section \ref{section_approximation_estimates}, we show that the averaged system accurately describes the original system by establishing error estimates on the approximation. This is done by providing precise error estimates that allow us to extend the results obtained in Section \ref{section_quasi_static_steer} to the full FHN system and prove Theorem \ref{thm_slow_steering1} and Theorem \ref{thm_slow_steering2}. 
\end{itemize}

Finally, in Section \ref{section_numerical_simulations}, we compare numerical results obtained for the FHN model to verify that both systems behave similarly concerning the studied properties, and we illustrate the dependence of the maximum permissible slopes and their relation to the stability of the systems.

\section{Systems with slowly varying coefficients}\label{section_abstract_steering_problem}

This section aims to show how to construct Lyapunov functions for specific non-autonomous nonlinear systems where the coefficients vary slowly over time. Such systems 
arise from averaging when the slope parameter is suitably small. The proofs in this section are adapted from the construction in \cite[Chapter 9]{khalil2013}, and we direct the reader to this reference for classic definitions and results mentioned below.

We will be working with matrix-valued functions, so it is helpful to identify the norms we will use.

\begin{definition}
For vectors $x\in\R^n$ we consider the Euclidean norm $\|x\|_2 = \left(\sum_{i=1}^n x_i^2\right)^{1/2}.$
 For matrices $A=\{a_{i,j}\}_{i,j} \in M_{m\times n}$ we consider the Frobenious norm $\|A\|_2 = \left(\sum_{i=1}^m \sum_{j=1}^{n} a_{i,j}^2\right)^{1/2}.$
\end{definition}

The following lemma tells us how the stability of a time-frozen linear system can be used to obtain stability for certain non-autonomous systems, non-linear systems, and systems with time-dependent source terms.

\begin{lemma}[Quasi-static steering]\label{lemma_lyapunov_nonlinear}
Let us consider the linear system
\begin{equation}\label{linear_system_lemma_quasistatic_steer}
\dot{x} = A(\alpha) x,
\end{equation}
where $A(\alpha): [0,1] \to M_{n\times n}(\R)$ is a matrix-valued function which is uniformly Hurwitz for $\alpha \in [0,1]$. Thus, it is known that the system \eqref{linear_system_lemma_quasistatic_steer} admits $L(\alpha,x) = x^T P(\alpha) x$ as a Lyapunov function where $P(\alpha)$ is given by the unique solution of the equation
\begin{equation}\label{riccati_slope_freq_lemma}
P(\alpha) A(\alpha) + A(\alpha)^T P(\alpha) = -I,
\end{equation}
and that there exist positive constants $c_1$ and $c_2$, independent of $\alpha \in [0,1]$, such that the function $L(\alpha,x)$ satisfies
\begin{equation*}
c_1 \|x\|_2^2\leq L(\alpha,x)\leq c_2 \|x\|_2^2.
\end{equation*}
Suppose that  $B: \R^n \times \R^+ \to \R^n$ and $C: \R^+ \to \R^n$ 
are two vector-valued functions
satisfying 
\begin{itemize}
\item[(i)] $\sup_{\alpha\in [0,1]}\|B(W,\alpha)\|_2 \leq f(\|W\|_2) \|W\|_2^2$, for some continuous monotonous increasing function $f$, and
\item[(ii)] $\sup_{\alpha\in[0,1]}\|C(\alpha)\|_2 \leq c_3$, for some positive constant $c_3$.
\end{itemize}
Then, there exists $\lambda_0   = (\sup_{\alpha\in[0,1]} \|P'(\alpha)\|_2 )^{-1}>0$ such that for all $\lambda < \lambda_0$ the function $L(\lambda t, x)$ is a Lyapunov function for the system
\begin{equation}\label{system_Z_lemma}
\dot{Z} = A(\lambda t) Z + B(Z, \lambda t),\qquad (t\in[0,1/\lambda]).
\end{equation}
Additionally, if $W$ is the unique solution of the system
\begin{equation}\label{system_W_lemma}
\dot{W} = A(\lambda t) W + B(W,\lambda t) +C(\lambda t)\lambda,\quad  W(0) = W_0, \quad t \in [0,1/\lambda],
\end{equation}
and the constants $\hat{M} >0$ and $\theta >0$ satisfy
\[
-\left(1-\lambda \sup_{\alpha\in[0,1]} \|P'(\alpha)\|_2\right)  + 2f(\hat{M}) \sup_{\alpha\in [0,1]} \|P(\alpha)\|_2 \hat{M}+ \theta\sup_{\alpha\in [0,1]}\left\|P(\alpha)\right\|_2^2 \leq 0,
\]
then we have for all $M\leq  \hat{M}$ 
\[
\frac{c_2}{c_1} \|W_0\|_2^2 + \frac{c_3^2}{\theta c_1}\lambda \leq M \quad\implies \|W(t)\|_2^2 \leq \frac{c_2}{c_1} \|W_0\|_2^2 + \frac{c_3^2}{\theta c_1}\lambda, \quad \forall \,t \in [0,1/\lambda].
\]
\end{lemma}

\begin{proof}
We look at the derivatives along  trajectories. Let $W$ be the solution of \eqref{system_W_lemma} and consider the time derivative of $L(\lambda t,W) = W^T P(\lambda t) W$. Take $\lambda>0$  so that $\lambda\sup_{\alpha\in [0,1]}\|P'(\alpha)\|_2 < 1 $, then we get for $t\leq \frac{1}{\lambda}$ 
\begin{align*}
\frac{d}{dt} L\left(\lambda t,W\right) &= \lambda W^T P'(\lambda t)W + \left(A(\lambda t)W + B(W,\lambda t) + \lambda C(\lambda t) \right)^T P(\lambda t) W\\
&\hspace{5mm}+ W^T P(\lambda t) \left(A(\lambda t)W + B(W,\lambda t) + \lambda C(\lambda t) \right)\\
&= \lambda W^T P'(\lambda t)W 
+W^T A^T(\lambda t)P(\lambda t) W +W^T P(\lambda t) A(\lambda t) W\\
&\hspace{5mm}+B(W,\lambda t)^T P(\lambda t) W +W^T P(t) B(W,\lambda t)\\
&\hspace{5mm}+\lambda  C(\lambda t)^T P(\lambda t) W +\lambda W^T P(\lambda t) 
 C(\lambda  t) .
\end{align*}
Thanks to equation \eqref{riccati_slope_freq_lemma} we have
\begin{equation*}
\lambda W^T P'(\lambda t)W + W^T A(\lambda t)^T P(\lambda t)W+ W^T P(\lambda t) A(\lambda t) W
\leq -\left(1-\lambda \sup_{\alpha\in [0,1]}\left\|P'(\alpha)\right\|_2\right) \left\|W\right\|^2.
\end{equation*}
For the nonlinear term $B(W,t)$ we use assumption (i) to obtain
\begin{equation*}
\left|W^T P(\lambda t)B(W,\lambda t)\right|\leq f(\|W\|_2) \sup_{\alpha\in [0,1]}\|P(\alpha)\|_2 \|W\|_2^3.
\end{equation*}
Putting everything together we get
\begin{align*}
\frac{d}{dt} L\left(\lambda t,W\right) &\leq -\left(1-\lambda \sup_{\alpha\in [0,1]}\|P'(\alpha)\|_2\right) \left\|W\right\|^2 \\
&\hspace{1cm}+ \left|B(W,\lambda t)^T P(\lambda t) W +W^T P(\lambda t)B(W,\lambda t)\right|\\
&\hspace{1cm}+\lambda \left|C(\lambda t)^T P(\lambda t) W + W^T P(\lambda t)C(\lambda t)\right|\\
&\leq -\left(1-\lambda \sup_{\alpha\in [0,1]}\|P'(\alpha)\|_2\right)  \|W\|_2^2 + 2 f(\|W\|_2) \sup_{\alpha\in [0,1]}\|P(\alpha)\|_2 \|W\|_2^3\\
&\hspace{1cm}+ 2\lambda \|C(\lambda t)\|_2 \left\|P(\lambda t) W\right\|_2\\
&\leq -\left(1-\lambda \sup_{\alpha\in [0,1]}\|P'(\alpha)\|_2\right) \|W\|_2^2 + 2 f(\|W\|_2)\sup_{\alpha\in [0,1]}\|P(\alpha)\|_2 \|W\|_2^3\\
&\hspace{1cm}+ \frac{1}{\theta} \lambda^2 \|C(\lambda t)\|^2  + \theta \sup_{\alpha\in[0,1]}\left\|P(\alpha)\right\|_2^2  \left\|W\right\|_2^2.
\end{align*}
Next, we choose $\hat{M}$ and $\theta$ such that
\begin{equation}\label{estimate_exponential_stability}
-\left(1-\lambda \sup_{\alpha\in[0,1]} \|P'(\alpha)\|_2 \right)  + 2f(\hat{M}) \sup_{\alpha\in [0,1]} \|P(\alpha)\|_2 \hat{M}+ \theta\sup_{\alpha\in [0,1]}\left\|P(\alpha)\right\|_2^2 \leq 0,
\end{equation}
which can always be done since the first term is strictly negative. We conclude that given any $M\leq \hat{M}$ in the ball $\|W \|_2 \leq M$ we have the estimate
\[
\frac{d}{dt} L(\lambda t , W ) \leq \frac{c_3^2 \lambda^2 }{\theta}.
\]
Integrating in $[0,t]$ we get
\[
L(\lambda t, W) \leq L(0,W(0)) + \frac{c_3^2 \lambda^2}{\theta} t,\]
and taking supremum in $t \in [0, 1/\lambda]$ we obtain
\[
\sup_{t\in [0,1/\lambda]}L(\lambda t, W) \leq L(0,W(0)) + \frac{c_3^2\lambda}{\theta}.
\]
We can clean up the estimate using the upper and lower bound in $L(\lambda t,W)$ given by \eqref{riccati_slope_freq_lemma} where the constants $c_1$ and $c_2$ are independent of $\alpha\in[0,1]$ via the same argument as in \cite[Lemma 9.9]{khalil2013}
\[
\sup_{t\in[0,1/\lambda]}\|W\|_2^2 \leq \frac{c_2}{c_1}\|W(0)\|_2^2 + \frac{c_3^2}{\theta c_1} \lambda.
\]
Finally, we can close this estimate by requiring
\[
\frac{c_2}{c_1}\|W(0)\|_2^2 + \frac{c_3^2}{\theta c_1} \lambda \leq M.
\]
Note that in the case $c_3 = 0$, it can be improved to exponential stability at the origin by requiring a strictly negative upper bound in \eqref{estimate_exponential_stability}, which implies that $L(\lambda t,W)$ is a Lyapunov function for \eqref{system_Z_lemma}. This concludes the proof of Lemma \ref{lemma_lyapunov_nonlinear}.
\end{proof}

For our application, it is desirable to understand how steep the ramp used on the time-varying term can be, which, thanks to the previous lemma, we know is related to the Lyapunov function via 
\[
\lambda_0 = \left(\sup_{\alpha\in[0,1]} \|P'(\alpha)\|_2\right)^{-1}.
\]
More specifically, we would like to conclude that under certain conditions, the derivative of the Lyapunov matrix is proportional to some quantity related to the Hurwitz condition. The following result provides the desired estimate by comparing the derivative of the matrix $P$, the solution of the Riccati equation \eqref{riccati_slope_freq_lemma} with the determinant and the trace of the matrix $A$ by assuming some structure on the coefficients of the linearization.

\begin{lemma}[Derivative estimates for the derivative of Lyapunov function]\label{derivative_lyapunov}
Consider the matrix-valued function $A(\alpha): [0,1] \to M_{2\times 2}(\R)$ of the form
\begin{equation*}
A(\alpha) = \left(\begin{array}{cc}
a_{11}(\alpha) & a_{12}\\
a_{21} & a_{22}
\end{array}\right),
\end{equation*}
for some real-valued function $a_{11}: [0,1] \to \R$ and constants $a_{12}$, $a_{21}$, $a_{22}$. Additionally, suppose the following
\begin{enumerate}[label=(\roman*)]
    \item $a_{22} <0$,\label{condition_negative_lemma_lyapunov}
    \item $\det A(\alpha) >0$, $\forall\alpha\in [0,1]$, \label{condition_det_lemma_lyapunov}
    \item $\tr A(\alpha)< 0$, $\forall\alpha\in [0,1]$, \label{condition_trace_lemma_lyapunov}
\end{enumerate}
and consider the matrix $P(\alpha)$ solution of the Riccatti equation
\begin{equation}\label{matrix_riccatti_lemma}
P(\alpha) A(\alpha) + A(\alpha)^T P(\alpha) = -I.
\end{equation}
Then, we have the following estimate on the derivative of $P(\alpha)$
\[
\frac{|a_{11}'|}{2 \tr(A)^2} \leq \|P'(\alpha)\|_2 \leq C_0 |a_{11}'| \frac{\|A\|_2^4}{(\tr(A) \det(A))^2},
\]
for some universal constant $C_0>0$ that does not depend on the coefficients of the matrix $A$.
\end{lemma}
\begin{proof}
Conditions \ref{condition_det_lemma_lyapunov} and \ref{condition_trace_lemma_lyapunov} guarantee that the Ricatti equation \eqref{matrix_riccatti_lemma} has a unique solution. Let $P =\left(\begin{array}{cc} p & q\\ q & r\end{array}\right)$ then we can solve explicitly to find
\begin{equation}\label{explicit_solution_riccatti}
p = - \frac{(a_{22}^2+a_{21}^2 + \det(A)) }{2\tr(A) \det(A)} ,\quad 
q= \frac{(a_{12}a_{22}+ a_{21}a_{11})}{2\tr(A) \det(A)}, \quad 
r=  -\frac{(a_{12}^2 + a_{11}^2+\det(A))}{2\tr(A) \det(A)}.
\end{equation}

We only need to find an adequate lower bound for one of the coefficients to find a lower bound for the norm $\|P'(\alpha)\|_2$. Taking the derivative of the coefficient $p$ we get
\begin{align*}
p' &= \frac{(a_{22}^2+a_{21}^2)}{2} \frac{a_{11}' \det(A) + a_{22} a_{11}' \tr(A)}{(\tr(A) \det(A))^2} +\frac{a_{11}'}{2 (\tr(A))^2}\\
&= \frac{a_{11}'}{2} \frac{(a_{22}^2+a_{21}^2)\det(A) + (a_{22}^2+a_{21}^2) a_{22}  \tr(A) + (\det(A))^2}{(\tr(A) \det(A))^2}.
\end{align*}
Note that because of our assumptions, all the terms in the numerator have the same sign. Then we can bound
\begin{equation*}
\|P'(\alpha)\|_2 = \left((p')^2 + 2(q')^2 + (r')^2\right)^{1/2} \geq |p'| \geq \frac{|a_{11}'|}{2(\tr(A))^2},
\end{equation*}
which concludes the proof of the lower bound. To get the upper bound, we obtain an upper bound of the numerator in terms of the coefficients of the matrix $A$. The constant $C_0>0$ is related to the number of terms in the numerator of the derivative of the matrix $A$.
\end{proof}

The relevance of this lemma is to provide an explicit estimate for the small parameter given by Theorem \ref{lemma_lyapunov_nonlinear}; this will allow us to compare the slope with quantities related to the stability of the system.

\section{Partially Averaged Systems}\label{section_quasi_static_steer}

In this section, we introduce the partially averaged system (PAS) as a tool to study systems with highly oscillatory terms and explain how the steering result from Section \ref{section_abstract_steering_problem} can be applied to the averaged system.

\subsection{Partial averaging as a tool to understand HFBS}\label{subsection_partial_averaging}

Averaging is a natural approach to studying systems with highly oscillatory sources. We consider a version where the averaging window moves over time, which allows us to feel the effects of other time-varying terms that change slowly over time. This technique is sometimes called Partial Averaging \cite{cerpa_partially_2022, cerpaApproximationStabilityResults2023a}.

\begin{definition}[Partially averaged system]
The partial averaging of system \eqref{FHN_hf_intro} with solution $(v(t),w(t))$ and input current \eqref{input_current_slope_general1} is given by the following initial value problem
\begin{equation}\label{PAS_slope_freq}
\left\{
\begin{array}{rll}
\dot{V}&= (1 - S(\lambda t)^2 \rho^2 /2 )V-V^3/3  -W &, \quad t > 0,\\
\dot{W}&= \varepsilon(V- \gamma W + \beta)&, \quad t > 0,\\
V(0) &= v(0),\quad W(0) = w(0),
\end{array}\right.
\end{equation}
where $S(t)$ is a ramp function given by \eqref{cutoff}. Similarly, the corresponding partially averaging of system \eqref{FHN_hf_intro} with solution $(v(t),w(t))$ and with input current \eqref{input_current_slope_general2} is given by the following initial value problem
\begin{equation}\label{FHN_pas}
\left\{
\begin{array}{rll}
\dot{V}&= (1-\rho^2/2)V -V^3/3 - W + S(\delta t) I_0&,   \quad t >0, \\
\dot{W}&= \varepsilon(V- \gamma W + \beta)&,  \quad t >0,\\
V(0) &= v(0),\quad W(0) = w(0).
\end{array}\right.
\end{equation}
The derivation of both systems is presented in Appendix \ref{appendix_derivation_pas_slope_HF} and \ref{appendix_derivation_pas_slope_DC}. 
\end{definition}

The main advantage of this formulation is that it transforms the problem with a highly oscillatory input into another one where a single coefficient varies slowly over time. This fact will be used to obtain our results.

\begin{definition}
\label{definition_condition_parameters}
The following conditions summarize the framework used to establish the results in the paper.
\begin{enumerate}[label={\bf C\arabic*}.,ref=C\arabic*]

\item \label{conditionC1} For fixed values of $(\varepsilon, \beta, \gamma, \rho)$, with  $\varepsilon,\beta,\gamma >0$, and every $\alpha\in[0,1] $, consider the system 
\begin{equation}\label{parameters1}
\begin{cases}
\dot{v} &= (1-\alpha^2 \rho^2/2) v - v^3/3 - w,  \\
\dot{w} &= \varepsilon \left( v-\gamma w + \beta \right).
\end{cases}
\end{equation}
System \eqref{parameters1} has a unique stable equilibrium and if $(v_0,w_0) = (v_0(\alpha),w_0(\alpha) )$ is the unique equilibrium of \eqref{parameters1}, there exists a constant $c >0$ such that for all $\alpha \in [0,1]$ 
\begin{equation}\label{stronger_stability_1}
v_0(\alpha)^2 - 1 + \alpha^2\rho^2/2 > c.
\end{equation}
\item \label{conditionC2} For fixed values of $(\varepsilon, \beta, \gamma, \rho, I_0)$, with  $\varepsilon,\beta,\gamma >0$, $I_0 \leq \beta/\gamma$,  and every $\sigma\in[0,1] $ consider the system  
\begin{equation}\label{parameters2}
\begin{cases}
\dot{v}  &= (1-\rho^2/2) v - v^3/3 - w + \sigma I_0,  \\
\dot{w} &= \varepsilon\left( v-\gamma w + \beta \right).
\end{cases}
\end{equation}
System \eqref{parameters2} has a unique stable equilibrium and if $(v_0,w_0) = (v_0(\sigma),w_0(\sigma) )$ is the unique equilibrium of \eqref{parameters2}, there exists a constant $c >0$ such that for all $\sigma \in [0,1]$ 
\begin{equation}\label{stronger_stability_2}
v_0(\sigma)^2 - 1 + \rho^2/2 > c.
\end{equation}
\item \label{conditionC3} Fix values of $(\varepsilon, \beta,\gamma,\rho)$, with $\varepsilon, \beta >0,  \gamma >1$. For $\lambda>0, \omega>0$, let $(v,w)$ be the unique solution of \eqref{FHN_hf_intro} with input current \eqref{input_current_slope_general1} and initial data $(v_0,w_0)$ given by \eqref{definition_P0}, and let $(V,W)$ be the solution of the partially averaged system \eqref{PAS_slope_freq} with the same initial data $(v_0,w_0)$. Let $(E_v, E_w) = (v -V - (\rho \lambda t) \sin(\omega t) , w - W)$ be the approximation error. There exists $M>0$ such that for all $\lambda \in (0,M]$
\[
\lim_{\omega\to\infty} \sup_{t\in[0,1/\lambda]} \|(E_v(t),E_w(t) )\|_2=0.
\]
\item \label{conditionC4} Fix $(\varepsilon,  \beta,\gamma, \rho, I_0)$, with $\varepsilon, \beta >0, \gamma>1$. For $\delta > 0$,  $\omega >0$ let $(v,w)$ be the solution of \eqref{FHN_hf_intro} with input current \eqref{input_current_slope_general2} and initial data $(v_1,w_1)$ given by \eqref{definition_P1}, and let $(V,W)$ be the solution of the Partially averaged system \eqref{FHN_pas} with the same initial data $(v_1,w_1)$. Let $(E_v, E_w) = (v -V - \rho  \sin(\omega t) , w - W)$ be the approximation error. There exists $M>0$ such that for all $\delta \in (0,M]$
\[
\lim_{\omega\to\infty} \sup_{t\in[0,1/\delta]} \|(E_v(t),E_w(t) )\|_2=0.
\]
\end{enumerate}
\end{definition}
\begin{definition}\label{defi_hypotheses}
The following hypotheses guarantee that the conditions above in Definition \ref{definition_condition_parameters} are satisfied.
\begin{enumerate}[(H1)] 
\item \label{condition_parameters} For fixed values of $(\varepsilon, \beta,\gamma,\rho , I_0, \alpha, \sigma)$, with $\varepsilon, \beta, \gamma >0$, $I_0 \leq  \beta/\gamma$ consider
\begin{equation*}
-\left(1-\alpha^2\rho^2/2-\frac{1}{\gamma}\right)^3  +\frac{1}{4} \left( 3  \frac{\beta}{\gamma}  -3 \sigma I_0\right) ^2> 0. 
\end{equation*}
\item \label{condition_parameters1}For fixed values of $(\varepsilon,  \beta,\gamma,\rho , I_0, \alpha, \sigma)$, with $\varepsilon, \beta, \gamma >0$, $I_0 \leq \beta/\gamma$ consider
\begin{equation*}
M = 1-\alpha^2 \rho^2/2 \leq 0,  \text{ or }\quad 
M^{3/2} + 3 \beta/\gamma -3 \sigma I_0 >3(1-\alpha^2 \rho^2/2 - 1/\gamma)M^{1/2}. 
\end{equation*}
\item \label{condition_approximation_HF} Fix values of $(\varepsilon, \beta,\gamma,\rho , I_0, \lambda, \omega)$, with $\varepsilon, \beta, \lambda, \omega >0$, and $\gamma >1$. Let $(V,W)$ be the solution of system \eqref{PAS_slope_freq}, consider the condition
\begin{multline*}
0<\frac{1}{\gamma}\exp\left(\frac{\rho^2}{2\omega} + \frac{4 |\rho|}{\omega} \max_{s\in[0,1/\lambda]}|V(s)|\right)
\left(\min_{s\in[0,1/\lambda]}\left(|V(s)|^2 +(\lambda s \rho)^2/2 -1\right) -\frac{1}{\omega} g(\rho,\lambda)\right)^{-1} <\frac{1}{2},
\end{multline*}
where $g(\rho,\lambda) = \frac{\rho^2 \lambda}{2} + 2|\rho| \lambda\max_{s\in[0,1/\lambda]}|V(s)| + 2 |\rho| \max_{s\in[0,1/\lambda]}|\dot{V}(s)|$.
\item \label{condition_approximation_DC} For fixed values of $(\varepsilon,  \beta,\gamma,\rho , I_0, \delta, \omega)$, $\varepsilon, \beta, \delta, \omega >0$, and $\gamma >1$. Let $(V,W)$ be the solution of system \eqref{FHN_pas}, consider the condition
\begin{multline*}
0< \frac{1}{\gamma}\exp\left(\frac{\rho^2}{2\omega} + \frac{4|\rho|}{\omega} \max_{s\in[0,1/\delta]}|V(s)|\right)\left(\min_{s\in[0,1/\delta]}\left(|V(s)|^2-1+\rho^2/2\right)  - \frac{ |\rho|}{\omega}\max_{s\in [0, 1/\delta}|\dot{V}(s)|  \right)^{-1}<\frac{1}{2}.
\end{multline*}
\end{enumerate}
\end{definition}

\begin{proposition}\label{proposition_parameters} The conditions introduced in Definition \ref{definition_condition_parameters} can be verified via the hypotheses in Definition \ref{defi_hypotheses} using the following:
\begin{enumerate}[(i)]
\item  Condition \ref{conditionC1} is equivalent to require hypotheses \ref{condition_parameters} and \ref{condition_parameters1} for all $\alpha\in [0,1]$ and $\sigma= 0$ (it is independent of $I_0$ when $\sigma = 0$).

\item Condition \ref{conditionC2} is equivalent to require hypotheses \ref{condition_parameters} and \ref{condition_parameters1} for   $\alpha=1$ and for all $\sigma\in[0,1]$.
\item Condition \ref{conditionC3} holds if Condition \ref{conditionC1} holds and
\begin{enumerate}[label={\bf C3*}.,ref=C3*]
\item \label{conditionC3b} 
There exists $M, \omega_0>0 \text{ such that for all } 0<\lambda \leq M, \omega\geq\omega_0$ the hypothesis \ref{condition_approximation_HF} is satisfied.
\end{enumerate}
\item Condition \ref{conditionC4} holds if Condition \ref{conditionC2} holds and 
\begin{enumerate}[label={\bf C4*}.,ref=C4*]
\item \label{conditionC4b}
There exists $M, \omega_0>0 \text{ such that for all }0<\delta \leq M, \omega\geq\omega_0$ the hypothesis \ref{condition_approximation_DC} is satisfied.
\end{enumerate}

\end{enumerate}
\end{proposition}

The proofs of (i) and (ii) in Proposition \ref{proposition_parameters} are presented in Appendix \ref{app_subsection_choice_parameters}. The proofs of (iii) and (iv) are given in Propositions \ref{proposition_error_estimates_slope_HF} and \ref{proposition_error_estimates_slope_DC} using Lemma \ref{lemma_linear_error_abs}.

We also show that the hypotheses {\it \ref{condition_parameters}-\ref{condition_approximation_DC}} are satisfied by a non-empty set of parameters.
\begin{proposition}\label{proposition_easy_conditions}
Fix $\varepsilon>0$, $I_0, \rho\in \R$. There exists $\gamma_0, h_0>0$ (depending on $\varepsilon, I_0, \rho$) such that if 
\begin{equation}\label{easy_condition}
\gamma > \gamma_0, \quad \beta/\gamma > h_0,
\end{equation}
then for all $\alpha, \sigma \in [0,1] $ hypotheses \ref{condition_parameters}, \ref{condition_parameters1} in Definition \ref{defi_hypotheses} are satisfied.

If in addition we assume that $\rho^2 < 4/3$ and $I_0 \geq 0$, then there exists $\mu^*$, $\lambda^*$, $\omega_0$, $\gamma_0$, $h_0>0$ such that if Condition \eqref{easy_condition} is satisfied, $0<\lambda < \lambda^*$, $\omega \geq \omega_0$ and
\begin{equation}\label{easy_condition_initial}
\|(V(0) - v^*, W(0)-w^* )\|_2 \leq \mu^*,
\end{equation}
where $(V,W)$ is the solution of \eqref{PAS_slope_freq} and  $(v^*,w^*)=(v_0,w_0)$ given by \eqref{definition_P0}, then hypothesis \ref{condition_approximation_HF} in Definition \ref{defi_hypotheses} is satisfied. Alternatively, under the same assumptions, there exists $\mu^*$, $\delta^*$, $\omega_0$, $\gamma_0$, $h_0>0$ such that if Condition \eqref{easy_condition} is satisfied, $0<\delta < \delta^*$, $\omega \geq \omega_0$ and \eqref{easy_condition_initial} is satisfied with $(V,W)$ being the solution of \eqref{FHN_pas} and  $(v^*,w^*)=(v_1,w_1)$ given by \eqref{definition_P1}, then hypothesis \ref{condition_approximation_DC} in Definition \ref{defi_hypotheses} is satisfied.
\end{proposition}
The proof is presented in Appendix \ref{appendix_proposition_easy_conditions}.

\subsection{Steering problem 1: The slope in the high-frequency term}\label{section_thm_slope_high_freq}

Since our goal is to approximate the solution of the averaged system \eqref{PAS_slope_freq} by the solution of the algebraic system \eqref{PAS_slope_freq_instantaneous}, it is crucial to have a good understanding of how the solutions of \eqref{PAS_slope_freq_instantaneous} behave. The following Lemma gives this.

\begin{lemma}[Properties of the algebraic system]\label{properties_algebraic_system_HF}
Suppose that $(\varepsilon, \beta, \gamma, \rho)$ satisfy Condition \ref{conditionC1} in Definition \ref{definition_condition_parameters}. Then for each $t\geq 0$ the solution $(X_1(t), Y_1(t))$ of the algebraic system \eqref{PAS_slope_freq_instantaneous} is unique and satisfies the following:
\begin{itemize}
\item[(i)] For each $t \geq 0$
\begin{equation}\label{uniform_estimate_X1}
X_1(t)^2 > 1- S(\lambda t)^2 \rho^2/2. 
\end{equation}
\item[(ii)] $X_1(t)$ and $Y_1(t)$ are continuous for all $t\geq 0$ and differentiable for $t\in (0, \frac{1}{\lambda})$.
\item[(iii)] $X_1(t)$  and $Y_1(t)$ are right differentiable at $t=0$ and left differentiable at $t= \frac{1}{\lambda}$.
\end{itemize}
\end{lemma}
\begin{proof}
First we note that for each $\alpha \in [0,1]$ the equation \eqref{PAS_slope_freq_instantaneous_alpha}  for $(\tilde{X}_1(\alpha),\tilde{Y}_1(\alpha))$ characterizes the equilibrium points of system \eqref{parameters1} and therefore Condition \ref{conditionC1} guarantees its uniqueness and implies that for each $\alpha\in[0,1]$ 
\begin{equation*}
\tilde{X}_1(\alpha)^2 >1-\alpha^2 \rho^2/2,
\end{equation*}
thus the change of variable  $\alpha = \lambda t$ gives us \eqref{uniform_estimate_X1} for $t\in [0,1/\lambda]$. For $t > 1/\lambda$ we use that $X_1(t) = X_1(1/\lambda)$.

For (ii), because of the inverse function theorem, the solution of \eqref{PAS_slope_freq_instantaneous_alpha} depends continuously differentiable on $t$ in an open interval if the Jacobian never vanishes. We can verify this directly,
\begin{equation*}
J(\alpha)= \det\left(\begin{array}{cc}
1-\alpha^2\rho^2/2 - \tilde{X}_1(\alpha)^2 & -1\\
1 & -\gamma
\end{array}\right) = \gamma \left(\tilde{X}_1(\alpha)^2 -1 +\alpha^2 \rho^2/2\right) + 1,
\end{equation*}
which is strictly positive for $\alpha\in (0,1)$ because of (i). The continuity of the solution at $\alpha=0$ and $\alpha= 1$ is a consequence of the uniqueness assumption in Condition \ref{conditionC1}.

For (iii) we use that the derivatives $\frac{d}{d\alpha}\tilde{X}_1$ and $\frac{d}{d\alpha}\tilde{Y}_1$ can be computed for all $\alpha\in (0,1)$ using implicit differentiation in \eqref{PAS_slope_freq_instantaneous} 
\begin{equation}\label{difference_quasi_static_2}
\frac{d}{d\alpha} \tilde{X}_1(\alpha) = \frac{-\alpha^2 \rho^2 \tilde{X}_1(\alpha)}{1-\alpha^2 \rho^2/2 -\tilde{X}_1(\alpha)^2- 1/\gamma} = \frac{ \gamma \alpha^2 \rho^2 \tilde{X}_1(\alpha)}{J(\alpha)},\quad \frac{d}{d\alpha}\tilde{Y}_1 = \frac{1}{\gamma}\frac{d}{d\alpha}\tilde{X}_1,
\end{equation}
with the denominator $J(\alpha)$ never vanishing because of (i). A change of variable gives the differentiability of $X_1(t)$ and $Y_1(t)$ for all $t \in (0,1/\lambda)$, which together with the continuity for $t \geq 0$ implies that $X_1(t)$ is right differentiable at $t=0$ and left differentiable at $t= \frac{1}{\lambda}$. 
This concludes the proof of Lemma \ref{PAS_slope_freq_instantaneous}.
\end{proof}

The following result shows us how the general steering result given by Lemma \ref{lemma_lyapunov_nonlinear} can be used to study the steering properties of the averaged system.

\begin{proposition}[Steering result for system \eqref{PAS_slope_freq}]\label{prop_slope_freq}
Let $(\varepsilon,\beta, \gamma,  \rho)$ satisfy Condition \ref{conditionC1} in Definition \ref{definition_condition_parameters} and let $(v_0,w_0)$ be given by \eqref{definition_P0}. There exists $\hat{\mu}>0$ such that if the initial condition satisfy
\[
\mu^2 = |V(0) - v_0|^2 + |W(0)-w_0|^2\leq \hat{\mu}^2,
\]
then there exists $\hat{\lambda} > 0$ independent of $\mu$, and constants $C = C(\varepsilon, \beta,  \gamma, \rho) > 0$ and  $C_1 = C_1(\varepsilon, \beta, \gamma, \rho, \hat{\lambda}) > 0$ such that for any $0< \lambda < \hat{\lambda}$ the corresponding solution $(V,W)\in C([0,1/\lambda]; \R^2)\cap C^1((0,1/\lambda); \R^2)$ of system \eqref{PAS_slope_freq} and the corresponding solution $(X_1(t), Y_1(t))\in C([0,1/\lambda]; \R^2) \cap C^1((0,1/\lambda);\R^2)$ of system \eqref{PAS_slope_freq_instantaneous} satisfy
\begin{equation}\label{estimate_difference_HF}
|V(t) - X_1(t)|^2 + |W(t)-Y_1(t)|^2 \leq C\mu^2 + C_1 \lambda , \quad t \in [0,1/\lambda].
\end{equation}
Additionally, we have the following estimate for the maximum permissible slope
\begin{equation}\label{estimate_slope1_HF}
\hat{\lambda} \leq \inf_{\alpha\in(0,1]} \frac{2}{\alpha \rho^2}\frac{(\tilde{X}_1(\alpha)^2 - 1 + \varepsilon\gamma + \alpha^2\rho^2/2)^2 ( \tilde{X}_1(\alpha)^2 -1 + 1/\gamma + \alpha^2\rho^2/2 )}{3 \tilde{X}_1(\alpha)^2 -1 + 1/\gamma + \alpha^2\rho^2/2 },
\end{equation}
and 
\begin{equation}\label{estimate_slope2_HF}
\hat{\lambda} \geq \frac{1}{C_0} \inf_{\alpha\in(0,1]} \frac{2}{\alpha \rho^2}\frac{(\tilde{X}_1(\alpha)^2 - 1 + \varepsilon\gamma + \alpha^2\rho^2/2)^2 }{ \left((\tilde{X}_1(\alpha)^2 - 1 + \alpha \rho^2/2 )^2 
+ 1+ \varepsilon^2 +\varepsilon^2\gamma^2 \right)^2 } \frac{( \tilde{X}_1(\alpha)^2 -1 + 1/\gamma + \alpha^2\rho^2/2 )^3}{\left(3 \tilde{X}_1(\alpha)^2 -1 + 1/\gamma + \alpha^2\rho^2/2 \right)},
\end{equation}
where the functions $(\tilde{X}_1(\alpha), \tilde{Y}_1(\alpha))\in C([0,1]; \R^2)$ are defined by  \eqref{PAS_slope_freq_instantaneous_alpha} and $C_0 > 0$ is a universal constant that does not depend on the parameters of the system.
\end{proposition}

\begin{remark}
  Proposition \ref{prop_slope_freq} tells us that for $\mu >0$ and $\lambda>0$ small enough, System \eqref{PAS_slope_freq} does not generate action potentials for all $t \geq 0$.
\end{remark}

\begin{proof}
The idea of the proof is to show that for $\lambda>0$ small we can estimate the difference between the solution of \eqref{PAS_slope_freq} and the solution of the algebraic system \eqref{PAS_slope_freq_instantaneous} by applying Lemma \ref{lemma_lyapunov_nonlinear} to obtain estimate \eqref{estimate_difference_HF}.

We first look at what happens for time $0\leq t \leq \frac{1}{\lambda}$. Let $(X_1,Y_1)$ be the solution of \eqref{PAS_slope_freq_instantaneous}, and $(V,W)$ be the solution of \eqref{PAS_slope_freq}. We want to compare $(X_1, Y_1)$ with $(V,W)$, by looking at $(R_v, R_w) = (V-X_1,W-Y_1)$ for $0\leq t \leq \frac{1}{\lambda}$ which satisfy the system
\begin{equation}\label{difference_system_slope_freq}
\left\{
\begin{array}{rl}
\dot{R}_v &= (1-X_1^2-\lambda^2t^2\rho^2/2) R_v - X_1 R_v^2 -\frac{1}{3}R_v^3 - R_w-\dot{X_1},\\
\dot{R}_w &= \varepsilon\left(R_v -\gamma R_w\right) - \dot{Y_1},\\
R_v(0) &=V(0)-X_1(0), \quad R_w(0) = W(0)-Y_1(0).
\end{array}\right.
\end{equation}
We want to apply Lemma \ref{lemma_lyapunov_nonlinear} to argue that for $\lambda>0$ small, the Lyapunov function for the linear part with time-frozen coefficients can be used as a Lyapunov function for the full system \eqref{difference_system_slope_freq}. The linear part of \eqref{difference_system_slope_freq} is given by $\frac{d}{dt}Z = A(\lambda t) Z$, where
\begin{equation}\label{matrix_linear_slope_HF}
A(\alpha) = \left(\begin{array}{cc} 1-\tilde{X}_1(\alpha)^2-\alpha^2\rho^2/2 & -1\\ \varepsilon & - \varepsilon \gamma
\end{array}\right), 
\end{equation}
and $\tilde{X}_1(\alpha): [0,1]\to \R$ is the solution of \eqref{PAS_slope_freq_instantaneous_alpha}. Because of condition \ref{conditionC1}, Lemma \ref{properties_algebraic_system_HF} tells us that $\tr A(\alpha)<0 $ and $\det A(\alpha)>0$. Hence, $A(\alpha)$ is Hurwitz for $\alpha\in[0,1]$. Thus, for each fixed $\alpha$ we can find a Lyapunov function of the form $L(\alpha,x) = x^T P(\alpha) x$, where $P(\alpha)$ satisfies
\begin{equation}\label{riccati_slope_freq}
P(\alpha) A(\alpha) + A^T(\alpha) P(\alpha) = -I. 
\end{equation}
To apply Lemma \ref{lemma_lyapunov_nonlinear}  we choose
\begin{equation*}
B(W,\alpha) = \binom{-\tilde{X}_1(\alpha) W_1^2 - \frac{1}{3} W_1^3}{0}, \quad C(\alpha) = \binom{-\frac{d}{d\alpha} \tilde{X}_1(\alpha)}{-\frac{d}{d\alpha}\tilde{Y}_1(\alpha)},
\end{equation*}
which means that system \eqref{difference_system_slope_freq} can be written as \eqref{system_W_lemma}. Lemma \ref{properties_algebraic_system_HF} tell us that the solution $(\tilde{X}_1(\alpha), \tilde{Y}_1(\alpha))$ of \eqref{PAS_slope_freq_instantaneous_alpha} is continuous and therefore uniformly bounded for $\alpha \in [0,1]$. This implies that 
\begin{equation*}
\sup_{\alpha\in[0,1]}\|B(W,\alpha)\|_2 \leq \left(\sup_{\alpha\in[0,1]}|\tilde{X}_1(\alpha)|+ \frac{1}{3} \|W\|_2\right) \|W\|_2^2,
\end{equation*}
which gives us condition (i) in Lemma \ref{lemma_lyapunov_nonlinear}. For condition (ii), Lemma \ref{properties_algebraic_system_HF} tells us that the denominator in \eqref{difference_quasi_static_2} is uniformly bounded away from zero
\begin{equation*}
\tilde{X}_1(\alpha)^2 - 1  + \alpha^2 \rho^2/2 + 1/\gamma > 1/\gamma> 0 , \quad \forall \alpha\in [0,1],
\end{equation*}
and the numerator is bounded because $\tilde{X}_1(\alpha)$ is uniformly bounded, which gives us
\[
\left|\frac{d}{d\alpha}\tilde{X}_1(\alpha)\right|\leq a, \quad \left|\frac{d}{d\alpha}\tilde{Y}_1(\alpha) \right| \leq b ,
\]
for some constants $a$, $b \geq 0$, and therefore $\sup_{\alpha\in[0,1]}\|C(\alpha)\|_2 \leq  \sqrt{a^2 + b^2}.$  This gives condition (ii) in Lemma \ref{lemma_lyapunov_nonlinear}. Finally, we can apply Lemma \ref{lemma_lyapunov_nonlinear} to conclude there exists $\hat{\mu}>0$ and $\hat{\lambda} >0$ such that if $0<\lambda < \hat{\lambda}$  and
\[
\|W_0\|_2^2  = |V(0)-v_0|^2 + |W(0)-w_0|^2\leq \hat{\mu}^2,
\]
then
\[
\|W(t)\|_2^2  \leq C \|W_0\|_2^2 + C_1\lambda,\quad 0<t<\frac{1}{\lambda},
\]
for constants $C = C(\varepsilon, \beta,  \gamma, \rho) > 0$ and  $C_1 = C_1(\varepsilon, \beta, \gamma, \rho, \hat{\lambda}) > 0$, which gives us the first part of the proposition. Estimates \eqref{estimate_slope1_HF} and \eqref{estimate_slope2_HF} are obtained by a direct application of Lemma \ref{derivative_lyapunov}. This ends the proof of Proposition \ref{prop_slope_freq}.
\end{proof}

\subsection{Steering problem 2: The slope problem for the DC term}\label{section_thm_slow_steer}

In this subsection, we prove the analogous of Proposition \ref{prop_slope_freq} by comparing the partially averaged system \eqref{FHN_pas} with the algebraic system  \eqref{instantaneous_system}.

\begin{lemma}[Properties of the algebraic system]\label{properties_algebraic_system_DC}
Suppose that $(\varepsilon, \beta, \gamma, \rho, I_0)$ satisfy Condition \ref{conditionC2} in Definition \ref{definition_condition_parameters}. Then for each $t\geq 0$ the solution $(X_2(t), Y_2(t))$ of the algebraic system \eqref{instantaneous_system_alpha} is unique and satisfies the following:
\begin{itemize}
\item[(i)] For each $t \geq 0$
\[
X_2(t)^2 > 1 - \rho^2/2.
\]
\item[(ii)] $X_2(t)$ and $Y_2(t)$ are continuous for all $t\geq 0$ and differentiable for $t\in (0,1/\delta)$.
\item[(iii)] $X_2(t)$ and $Y_2(t)$ are right differentiable at $t=0$ and left differentiable at $t= 1/\delta$.
\end{itemize}
\end{lemma}
The proof of this fact is analogous to the proof of Lemma \ref{properties_algebraic_system_HF}.

\begin{proposition}\label{lemma_slow_steer}
Suppose that $(\varepsilon, \beta, \gamma, \rho, I_0)$ satisfy Condition \ref{conditionC2} in Definition \ref{definition_condition_parameters}. There exists $\hat{\mu}= \hat{\mu}(\rho,I_0)>0$ such that if initial condition satisfying
\[
\mu^2 = |V(0) - v_1 |^2+ |W(0) - w_1|^2 \leq  \hat{\mu}^2,
\]
where $(v_1, w_1)$ is given by \eqref{definition_P1}. Then there exists $\hat{\delta}>0$ independent of $\mu > 0$, and constants $C = C(\varepsilon, \beta,  \gamma, \rho,I_0) > 0$ and  $C_1 = C_1(\varepsilon, \beta, \gamma, \rho,I_0,\hat{\delta}) > 0$, and  such that for any $0<\delta < \hat{\delta}$ the solution $(V,W)$ of system \eqref{FHN_pas} and the solution $(X_2(t),Y_2(t))$ of system \eqref{instantaneous_system} can be compared by
\begin{equation*}
|V(t) - X_2(t)|^2 + |W(t) - Y_2(t)|^2 \leq C \mu^2 +C_1 \delta, \quad t \in [0,1/\delta].
\end{equation*}

Additionally, we have the following estimate of the size of the maximum slope
\[
\hat{\delta} \leq \frac{1}{I_0}\inf_{\alpha\in[0,1] }\frac{(\tilde{X}_2(\alpha)^2 - 1 + 1/\gamma + \rho^2/2) (\tilde{X}_2(\alpha)^2-1+ \rho^2/2 + \varepsilon \gamma)^2}{\tilde{X}_2(\alpha) }, 
\]
and 
\begin{equation*}
\hat{\delta} \geq \frac{1}{2C_0 I_0}\inf_{\alpha\in[0,1] }\frac{ (\tilde{X}_2(\alpha)^2 - 1 + 1/\gamma + \rho^2/2) (\tilde{X}_2(\alpha)^2-1+ \rho^2/2 + \varepsilon \gamma)^2  ( \varepsilon \gamma (\tilde{X}_2(\alpha)^2 -1 + \rho^2/2) + \varepsilon )^2}{ \tilde{X}_2(\alpha) \left((\tilde{X}_2(\alpha)^2 - 1 + \rho^2/2)^2 +1 +\varepsilon^2 + \varepsilon^2\gamma^2\right)^2}
 ,
\end{equation*}
where the functions $(\tilde{X}_2(\alpha),\tilde{Y}_2(\alpha))\in C([0,1];\R^2)$ are defined by \eqref{instantaneous_system} and $C_0$ is a universal constant that does not depend on the parameters of the system.

\end{proposition}

\begin{remark}
  This proposition is analogous to Proposition \ref{prop_slope_freq} and tells us that for $\mu >0$ and $\delta>0$ small enough, the system does not generate action potentials for all $t \geq 0$.
\end{remark}

\begin{proof}
Let $(V,W)$ be the unique solution of system \eqref{FHN_pas} and let $(X_2(t),Y_2(t)): [0,1/\delta] \to \R^2$ be the unique solution to \eqref{instantaneous_system}. Consider $(R_v, R_w) = (V-X_2,W-Y_2)$ which satisfy the system
\begin{equation}\label{difference_system}
\left\{
\begin{array}{rl}
\dot{R}_v &= (1-X_2^2-\rho^2/2) R_v - X_2 R_v^2 -\frac{1}{3}R_v^3 - R_w-\dot{X_2},\\
\dot{R}_w &= \varepsilon\left(R_v -\gamma R_w\right) - \dot{Y_2},\\
R_v(0) &= V(0)-X(0), \quad R_w(0) = W(0)-Y(0).
\end{array}\right.
\end{equation}
The derivatives $\dot{X}_2$ and $\dot{Y}_2$ are computed using implicit differentiation in \eqref{instantaneous_system_alpha}, which gives us
\begin{equation}\label{estimate_instantaneous_system}
\frac{d}{d\alpha}\tilde{X}_2 = \frac{I_0}{\tilde{X}_2(\alpha)^2-1+1/\gamma +\rho^2/2}, \quad \frac{d}{d\alpha}\tilde{Y}_2 = \frac{1}{\gamma}\frac{d}{d\alpha} \tilde{X}_2.
\end{equation}
As in the Proof of Proposition \ref{prop_slope_freq}, we want to apply Lemma \ref{lemma_lyapunov_nonlinear}. For this purpose, we look at the linear part of system \eqref{difference_system} given by $\dot{Z} = A(\delta t) Z$ where
\begin{equation}\label{coefficient_matrix_linear_DC_slope}
A(\alpha) = \left(\begin{array}{cc} 1-\tilde{X}_2(\alpha)^2-\rho^2/2 & -1\\ \varepsilon & -\varepsilon \gamma
\end{array} \right),\quad \alpha \in [0,1].
\end{equation}
Because of Condition \ref{conditionC2}, Lemma \ref{properties_algebraic_system_DC} tells us that $\tr(A(\alpha) )<0$, $\det A(\alpha)>0$ for $\alpha \in [0,1]$. This implies that for each fixed $\alpha\in [0,1]$ we can find a Lyapunov function of the form $L(\alpha,x) = x^T P(\alpha) x$, where $P(\alpha)$ satisfies
\begin{equation}\label{equation_P}
P(\alpha) A(\alpha) + A^T(\alpha) P(\tau) = -I.
\end{equation}
Next, by choosing
\[
B(W,\alpha) = \binom{-\tilde{X}_2(\alpha) W_1^2 - \frac{1}{3} W_1^3}{0}, \quad C(\alpha) = \binom{-\frac{d}{d\alpha}\tilde{X}_2(\alpha)  }{-\frac{d}{d\alpha}\tilde{Y}_2(\alpha)},
\]
we see that system \eqref{difference_system} can be written as \eqref{system_W_lemma}. By continuity, it is easy to see that solutions of \eqref{PAS_slope_freq_instantaneous} are uniformly bounded for $t \in [0,1/\delta]$. This implies that 
\[
\sup_{\alpha\in[0,1]}\|B(W,\alpha)\|_2 \leq \left(\sup_{\alpha\in[0,1]}|\tilde{X}_2(\alpha)|+ \frac{1}{3} \|W\|_2\right) \|W\|_2^2.
\]
Because of Condition \ref{conditionC2}, Lemma \ref{properties_algebraic_system_DC} tell us that
\[
\tilde{X}_2(\alpha)^2 -1 + \rho^2/2 + 1/\gamma >1/\gamma>0,\quad  \forall \alpha\in [0,1],
\]
we can apply \eqref{estimate_instantaneous_system} to obtain the bound $|\frac{d}{d\alpha}\tilde{X}_2(\alpha)|\leq a$,  $|\frac{d}{d\alpha}\tilde{Y}_2(\alpha)| \leq b$, for some constants $a$, $b \geq 0$. This gives condition (ii) in Lemma \ref{lemma_lyapunov_nonlinear},
\[
\sup_{\alpha\in[0,1]} \|C(\alpha)\|_2 \leq \sqrt{a^2 + b^2}.
\]
Finally, Lemma \ref{lemma_lyapunov_nonlinear} tell us there exists $\hat{\mu}>0$ and $\hat{\delta} >0$ such that if 
$\|W_0\|_2^2 \leq \hat{\mu}^2$ and $0<\delta < \hat{\delta}$ then
\[
\|W\|_2^2  \leq C \|W_0\|_2^2 + C_1\delta,
\]
for constants $C = C(\varepsilon, \beta,  \gamma, \rho) > 0$ and  $C_1 = C_1(\varepsilon, \beta, \gamma, \rho, \hat{\delta}) > 0$, which gives us the first part of the proposition. The estimates for the size of $\hat{\delta}$ are obtained by a direct application of Lemma \ref{derivative_lyapunov}. This concludes the proof of Proposition \ref{lemma_slow_steer}.
\end{proof}

\section{From PAS to FHN system: proofs of Theorems  \ref{thm_slow_steering1} and \ref{thm_slow_steering2}}\label{section_approximation_estimates}

In this section we show how estimates for the partial averaged system obtained in Subsection \ref{section_thm_slope_high_freq} and Subsection \ref{section_thm_slow_steer} can be used to obtain precise steering results for the FHN system \eqref{FHN_hf_intro} in the form of Theorems \ref{thm_slow_steering1} and \ref{thm_slow_steering2}. The strategy used is an extension of the approach used in \cite{cerpa_partially_2022} and \cite{cerpaApproximationStabilityResults2023a}, the main difference being that instead of using a uniform estimate in the linear term to get enough stability, we use an integral estimate that depends on the frequency $\omega$.

In Subsection \ref{subsec_approximation_abstract}, we study a somewhat general equation for the approximation error. In Subsection \ref{subsec_approximation_hf_term} and Subsection \ref{subsec_approximation_dc_term}, we verify that the previous approximation result applies to the problems under consideration. Finally, in Subsection \ref{subsec_proof_thm1} and Subsection \ref{subsec_proof_thm2},  we complete the proofs of Theorem \ref{thm_slow_steering1} and Theorem \ref{thm_slow_steering2}, respectively.

\subsection{The equation for the approximation error}\label{subsec_approximation_abstract}

\begin{definition}[General equation for approximation errors] Given $T>0$ and time depending functions $\varphi_\textit{slow}$, $\varphi_\textit{fast}$, $\varphi_2$, $H\!F_1$,  $H\!F_2$ $\in C([0,T]; \R)$ we consider the error functions $(E_v, E_w) \in C( [0,T]; \R) \cap C^1( (0,T); \R)$ given by the solution of the following nonlinear system
\begin{equation}\label{error_equation_abs}
\left\{
\begin{array}{ll}
\dot{E}_v  =  \left(\varphi_\textit{slow} + \varphi_\textit{fast}\right)E_v + \varphi_2 E_v^2  -\frac{1}{3}E_v^3 -E_w+ H\!F_1 &, 0<t<T,\\
\dot{E}_w = \varepsilon( E_v - \gamma E_w) + H\!F_2&, 0<t<T,\\
E_v(0)= 0,\quad E_w(0) = 0.
\end{array}\right.
\end{equation}
\end{definition}
Here we think of $\varphi_\textit{fast}$, $H\!F_1$, and $H\!F_2$ as rapidly oscillatory terms depending on a parameter $\omega$. 
We will obtain estimates for the error functions in the interval $[0,T]$ via a fixed point argument. The key argument is to formulate the fixed point so we can use a bound for the integral of the linear term and not necessarily a uniform bound.

\begin{lemma}[Estimate for the general error]\label{lemma_linear_error_abs}
Fix $T>0$, $\omega_0>0$ and the parameters $(\varepsilon, \beta, \gamma , \rho)$ with $\varepsilon >0$, $\beta >0$ and $\gamma > 1$. Suppose that the functions $\varphi_\textit{slow}, \varphi_\textit{fast}, \varphi_2, H\!F_1 , H\!F_2$   $\in C([0,T]; \R)$ of \eqref{error_equation_abs} satisfy the following conditions for all $t\in [0,T]$ and all $\omega \geq \omega_0$, 
\begin{enumerate}[label=(\roman*)]
\item\label{condition_lemma_1} $\frac{1}{\gamma} \sup\limits_{t\in [0, T]}\int_0^t e^{\int_\tau^t(\varphi_\textit{slow}(s) + \varphi_{\textit{fast}}(s))ds}d\tau <\frac{1}{2}$, 
\item\label{condition_lemma_2} $\sup\limits_{t\in [0, T]}\left|\int_0^t e^{\int_\tau^t\left( \varphi_\textit{slow}(s) + \varphi_\textit{fast}(s)\right)ds} H\!F_1(\tau) d\tau\right| \leq C_1 / \omega$,
\item\label{condition_lemma_3} $\sup\limits_{t\in [0, T]}\left|\int_0^t  e^{-\varepsilon\gamma (t-\tau)} H\!F_2(\tau) d\tau\right| \leq C_2 / \omega$,
\item\label{condition_lemma_4} $\sup\limits_{t\in [0, T]}|\varphi_2| \leq K_2$,
\end{enumerate}
where $C_1 = C_1(T,\omega_0)$, $C_2 = C_2(T,\omega_0)$ and $K_2=K_2(T)$ are positive constants. Then there exists $\omega^* \geq \omega_0$ such that the solution $(E_v, E_w)$ of \eqref{error_equation_abs}  satisfies the estimate
\begin{equation}\label{linear_estimate_general}
|E_v(t)|\leq C/\omega,\quad |E_w(t)| \leq C/\omega,\quad   \text{ for all }0<t<T,\quad  \omega \geq \omega^*, 
\end{equation}
 for some constant $C = C(C_1, C_2, \gamma)>0$.

\end{lemma}

\begin{proof}
Using integrating factors, we can write equation \eqref{error_equation_abs} as
\begin{equation*}
\begin{cases}
E_v  =  \int_0^t e^{\int_\tau^t\left( \varphi_\textit{slow} + \varphi_\textit{fast}\right)ds}\left(-E_w +\varphi_2 E_v^2 - \frac{1}{3}E_v^3 + H\!F_1\right) d\tau,\\
E_w = \int_0^t e^{-\varepsilon\gamma (t-\tau)}(\varepsilon E_v +H\!F_2)d\tau.
\end{cases}
\end{equation*}
Substituting the second equation in the first one, we get the equivalent formulation
\begin{equation}\label{integral_formulation}
\left\{
\begin{array}{rl}
E_v  &=  \int_0^t e^{\int_\tau^t\left( \varphi_\textit{slow} + \varphi_\textit{fast}\right)ds}\Big(-\int_0^\tau e^{-\varepsilon\gamma (\tau-s)}(\varepsilon E_v(s) +H\!F_2(s))ds \\
&\hspace{1cm}+\varphi_2 E_v^2 - \frac{1}{3}E_v^3 +H\!F_1\Big) d\tau,\\
E_w &= \int_0^t e^{-\varepsilon\gamma (t-\tau)}(\varepsilon E_v +H\!F_2)d\tau.
\end{array}
\right.
\end{equation}
Consider the map $\Gamma : C([0,T];\R) \to 
C([0,T];\R)$ defined by 
\begin{equation*}
\Gamma (f)(t)  =  \int_0^t e^{\int_\tau^t\left(  \varphi_\textit{slow} + \varphi_\textit{fast}\right)ds} \left( - \int_0^\tau e^{-\varepsilon\gamma (\tau-s)}(\varepsilon f(s) +H\!F_2(s))ds  + \varphi_2 f^2 - \frac{1}{3}f^3 + H\!F_1\right) d\tau.
\end{equation*}
Given $M_1>0$, let us consider the closed set $\Omega \subset C([0,T],\R^2)$ defined by
\[
\Omega = \left\{f\in C([0,T];\R): \sup\limits_{t\in [0, T]}|f| \leq M_1\right\}.
\]
In what follows, we denote by $\delta_1(\omega) = C_1/\omega$ and $\delta_2(\omega) = C_2/\omega$ the upper bounds in assumptions (ii) and (iii). We also define $M_1(\omega)=4\delta_1(\omega)+2\gamma\delta_2(\omega)$.\smallskip

\noindent
{\bf Step 1: }We will show that when $\omega$ is large enough and $M_1$ is small enough, then $\Gamma$ maps $\Omega$ into itself. Using the conditions from Lemma \ref{lemma_linear_error_abs} and assuming $\omega\geq \omega_0$, we estimate,
\begin{align*}
A &= \sup\limits_{t\in [0, T]}|\Gamma (f)| \\
&\leq \sup\limits_{t\in [0, T]}\int_0^t e^{\int_\tau^t\left(  \varphi_\textit{slow} + \varphi_\textit{fast}\right)ds}\\
&\hspace{1cm}\times\left| -\int_0^\tau e^{-\varepsilon\gamma (\tau-s)}(\varepsilon f(s) +H\!F_2(s))ds  + \varphi_2 f^2 - \frac{1}{3}f^3 + H\!F_1\right| d\tau \\
&\leq  \sup\limits_{t\in [0, T]}\int_0^t e^{\int_\tau^t\left(  \varphi_\textit{slow} + \varphi_\textit{fast}\right)ds}d\tau \sup\limits_{t\in [0, T]}\left|\int_0^t e^{-\varepsilon\gamma (t-\tau)}(\varepsilon f +H\!F_2)d\tau \right| \\
& \hspace{1cm}+ \sup\limits_{t\in [0, T]}\int_0^t e^{\int_\tau^t\left(  \varphi_\textit{slow} + \varphi_\textit{fast}\right)ds}d\tau\left(\sup\limits_{t\in [0, T]} |\varphi_2| \sup\limits_{t\in [0, T]}|f|^2 + \frac{1}{3} \sup\limits_{t\in [0, T]} |f|^3  \right)\\
& \hspace{1cm}+ \sup\limits_{t\in [0, T]}\left|\int_0^t e^{\int_\tau^t\left(  \varphi_\textit{slow} + \varphi_\textit{fast}\right)ds} H\!F_1 d\tau\right|\\
&\leq \int_0^t e^{\int_\tau^t\left(  \varphi_\textit{slow} + \varphi_\textit{fast}\right)ds}d\tau \left(\frac{1}{\gamma}M_1+\delta_2(\omega) + K_2 M_1^2 + \frac{1}{3} M_1^3\right) + 
\delta_1(\omega)\\
&\leq \frac{1}{2} M_1 + \frac{\gamma}{2} K_2 M_1^2 + \frac{\gamma}{6} M_1^3 + \delta_1 +\frac{\gamma}{2}\delta_2
= M_1 \left(\frac{3}{4} + \frac{\gamma}{2} K_2 M_1 + \frac{\gamma}{6} M_1^2\right).
\end{align*}
Let $h_1(K_2, \gamma):= \frac{3}{2}\left(-K_2+ \sqrt{K_2^2 + \frac{2}{3\gamma} }\right)$, and choose $\omega_1\geq \omega_0$ such that $M_1\leq h_1(K_2,\gamma)$ for all $\omega\geq \omega_1$. Then $\frac{\gamma}{2} K_2 M_1 + \frac{\gamma}{6}M_1^2 \leq \frac{1}{4}$ and 
\[
\sup_{t\in[0,T]}|\Gamma(f)| \leq M_1 \left(\frac{3}{4}  + \frac{\gamma}{2}K_2 M_1 + \frac{\gamma}{6}M_1^2\right)\leq M_1, \quad \text{ for all } \omega\geq\omega_1.
\]

\noindent
{\bf Step 2:} $\Gamma$ as a contraction mapping in $\Omega$. Let $f_1, f_2 \in \Omega$, then we estimate
\begin{align*}
B &= \sup\limits_{t\in [0, T]}\left|\Gamma(f_1) - \Gamma(f_2)\right|\\
&\leq \sup\limits_{t\in [0, T]} \int_0^t e^{\int_\tau^t\left( \varphi_\textit{slow} + \varphi_\textit{fast}\right)ds}\\
&\hspace{1cm}\times \Big|\int_0^\tau \varepsilon e^{-\varepsilon\gamma (\tau-s)}(f_1-f_2) ds +\varphi_2 ( f_1^2 -  f_2^2)-\frac{1}{3} (f_1^3 - f_2^3) \Big|d\tau\\ 
&\leq  \sup\limits_{t\in [0, T]}\int_0^t e^{\int_\tau^t\left( \varphi_\textit{slow} + \varphi_\textit{fast}\right)ds} d\tau \frac{1}{\gamma} \sup\limits_{t\in [0, T]}\Big|f_1-f_2\Big| \\
& + \sup\limits_{t\in [0, T]}\int_0^t e^{\int_\tau^t\left( \varphi_\textit{slow} + \varphi_\textit{fast}\right)ds}d\tau \sup\limits_{t\in [0, T]}|\varphi_2| \sup\limits_{t\in [0, T]}(|f_1| + |f_2|) \sup\limits_{t\in [0, T]}|f_1 -  f_2|\\
& +\frac{1}{3} \sup\limits_{t\in [0, T]}\int_0^t e^{\int_\tau^t\left( \varphi_\textit{slow} + \varphi_\textit{fast}\right)ds} d\tau  \sup\limits_{t\in [0, T]}(|f_1|^2 + |f_1 f_2| + |f_2|^2) \sup\limits_{t\in [0, T]}|f_1 - f_2|\\
&\leq \alpha(T) \sup\limits_{t\in [0, T]} |f_1 - f_2|,
\end{align*}
where 
\begin{align*}
\alpha(T) &= \sup\limits_{t\in [0, T]}\int_0^t e^{\int_\tau^t\left( \varphi_\textit{slow} + \varphi_\textit{fast}\right)ds} d\tau \left( \frac{1}{\gamma} + 2 K_2 M_1 + M_1^2 \right)\\
&\leq\frac{\gamma}{2} \left( \frac{1}{\gamma} + 2 K_2 M_1 + M_1^2 \right).
\end{align*}
Define $ h_2(K_2,\gamma):= -K_2 +\sqrt{K_2^2 + \frac{1}{\gamma} }$ and choose $\omega^*\geq \omega_1$ such that 
$M_1\leq h_2(K_2,\gamma)$ for all $\omega\geq \omega^*$. Then, $(2 K_2 M_1 + M_1^2)\leq \frac{1}{\gamma}$ and $\alpha(T)<1$ for all $\omega\geq \omega^*$.
In summary, for $\omega \geq \omega^*$  the map $\Gamma$ is a contraction mapping in $\Omega$ and therefore has a unique fixed point inside. To obtain estimate  \eqref{linear_estimate_general} it is enough to observe that $|E_v|  \leq M_1 = 4 (\delta_1 + \frac{\gamma}{2}\delta_2) = (4 C_1 + 2 \gamma C_2) /\omega$. Because of \eqref{integral_formulation} This can be also used to estimate $E_w$, we get     $|E_w| \leq \frac{1}{\gamma}M_1  + \delta_2 = (4C_1/\gamma + 3C_2 ) /\omega,$
for all $\omega \geq \omega^*$. This concludes the proof of Lemma \ref{lemma_linear_error_abs}.
\end{proof}

\subsection{Estimates for the approximation error: steering problem 1}\label{subsec_approximation_hf_term}

This section is devoted to obtaining an estimate for the approximation error, implying condition \ref{conditionC3}. This is done by using Lemma \ref{lemma_linear_error_abs} in the specific setting of the steering problem 1.  The equation for the error is obtained directly by subtracting equations \eqref{FHN_hf_intro} and \eqref{PAS_slope_freq} and using that $S(\lambda t) = \lambda t$ for $t \in [0,1/\lambda]$.

\begin{proposition}[Equation for the approximation error: steering problem 1]\label{defi_approximation_problem_HF}
Let $(v,w)$ be the unique solution of \eqref{FHN_hf_intro} with input current \eqref{input_current_slope_general1} and initial data $(v_0,w_0)$. Let $(V,W)$ be the solution of the partially averaged system \eqref{PAS_slope_freq} (which has the same initial data $(v_0,w_0)$). For $t \in [0,1/\lambda]$ , the approximation error functions $E_v = v -V - (\rho \lambda t) \sin(\omega t) $, $E_w = w - W$ satisfy 
\begin{equation}\label{error_equation_HF}
\left\{
\begin{array}{rll}
\dot{E}_v  &=  \left(\varphi_\textit{slow} + \varphi_\textit{fast}\right)E_v + \varphi_2 E_v^2  -\frac{1}{3}E_v^3 -E_w+ H\!F_1 &, 0<t<\frac{1}{\lambda},\\
\dot{E}_w &= \varepsilon( E_v - \gamma E_w) + H\!F_2&, 0<t<\frac{1}{\lambda},\\
E_v(0)&= 0,\quad E_w(0) = 0,
\end{array}\right.
\end{equation}
where 
\begin{align}
\varphi_\textit{slow} &= 1-V^2 - \frac{(\rho\lambda t)^2}{2}, \label{defi_varphi11_HF}\\
\varphi_\textit{fast} &= \frac{(\rho \lambda t)^2}{2} \cos(2\omega t) -2 (\rho \lambda t)V \sin(\omega t),\label{defi_varphi12_HF}\\
\varphi_2 &= -(\rho \lambda t \sin(\omega t) + V), \label{defi_varphi2_HF}\\
H\!F_1 &=-  (\rho \lambda t) V^2 \sin(\omega t) + \frac{ (\rho \lambda t)^2V}{2} \cos(2\omega t)   -\frac{1}{3}(\rho \lambda t)^3 \sin^3(\omega t) \label{defi_varphi3_HF} \\
&\hspace{1cm} + (\rho \lambda t)\sin(\omega t) - \rho \lambda \sin(\omega t), \notag \\
H\!F_2 &= \varepsilon (\rho \lambda t) \sin(\omega t)\label{defi_HF2_HF}.
\end{align}
\end{proposition}

The following estimate will be necessary to analyze the elements in equation \eqref{error_equation_HF}.
\begin{lemma}[Integral estimate for slowly varying functions]\label{integral_estimate_slow_varying} Let $T>0$, and $f \in C^1([0,T])$. Then we have
\begin{equation*}
\left\vert \int_a^t f(\tau) e^{i \omega \tau} d\tau\right\vert\leq   \frac{1}{\omega} \min\{(t-a) \omega,2\} \max_{\tau \in [a,t]}|f(\tau)|+ \frac{(t-a)}{\omega} \max_{\tau\in[a,t]} \vert f'(\tau)\vert .
\end{equation*}
\end{lemma}
\begin{proof} 
First, triangle inequality tells us $\left|\int_a^t f(\tau) e^{i\omega \tau}d\tau\right| \leq  (t-a ) \max\limits_{\tau\in [a,t]} |f(\tau)|$. Second, using integration by parts
\begin{equation*}
I(t) = \int_a^t e^{i \omega \tau}f(\tau) d\tau= \left. \frac{1}{i \omega }e^{i\omega \tau} f(\tau) \right|_{\tau=a}^{\tau = t} - \frac{1}{i\omega}\int_a^t  e^{i\omega \tau} f'(\tau) d\tau,
\end{equation*}
applying triangle inequality gives us $|I(t)|\leq \frac{2}{\omega} \max\limits_{\tau\in [a,t]} |f(\tau)| + \frac{t-a}{\omega} \max\limits_{\tau\in [a,t]} |f'(\tau)|$. Combining the estimates, we conclude the proof.
\end{proof}

The following proposition provides properties of system \eqref{PAS_slope_freq} that will be useful to check the hypothesis of Lemma \ref{lemma_linear_error_abs}.

\begin{proposition}\label{prop_properties_pas_HF} Let $(\varepsilon,\beta, \gamma,  \rho)$ satisfy Condition \ref{conditionC1} in Definition \ref{definition_condition_parameters}. Let $\hat{\lambda}, \hat{\mu}>0$ given by Proposition \ref{prop_slope_freq}. There exists $\lambda^* \leq \hat{\lambda}$ and $\mu^* \leq \hat{\mu}$ such that for all $0<\lambda\leq \lambda^*$, the corresponding solution $(V,W)$ of system \eqref{PAS_slope_freq} with initial condition satisfying
\begin{equation}\label{smallness_initial_data_PAS}
\|(V(0) - v_0, W(0)-w_0 )\|_2 \leq \mu^*,
\end{equation}
where $(v_0,w_0)$ is given by \eqref{definition_P0}, satisfies for some constants $C_1 = C_1(\beta,\gamma, \rho, \mu^*)>0$, $C_2=C_2(\beta,\gamma, \rho, \mu^*)>0$:
\begin{enumerate}[label=(\roman*)]
\item $\max\limits_{t\in [0,1/\lambda]} |V(t)| \leq C_1$,
\item $\max\limits_{t\in [0,1/\lambda]} |\dot{V}(t)| \leq C_2$,
\item $\min\limits_{t\in [0,1/\lambda]} (V(t)^2 + (\rho \lambda t )^2/2)>1 $.
\end{enumerate}
\end{proposition}

\begin{proof}
Consider the following general scheme. Let $f:[0,T]\to\R$ be a continuous function; let us construct a  family of nested positively invariant rectangles $R_k$, satisfying $\cup_{k\geq k_0} R_k = \R^2$, for the system 
\begin{equation}\label{positively_invariant_rectangles}
\begin{array}{rl}
\dot{V} &= f(t) V - V^3/3 -W,\\
\dot{W} &= \varepsilon(V-  \gamma W + \beta).
\end{array}
\end{equation}
Namely, we want to find $M,N>0$ such that the rectangle $R = [-M,M]\times [-N,N]$
satisfies that whenever $(V,W)$ touches the boundary of $R$, the vector $(\dot{V},\dot{W})$ points towards the interior of $R$, which will be the case if:
\begin{itemize}
    \item On $V =M$, $|W| \leq N$ we require $\dot{V} \leq \sup_{s\in[0,T]} |f(s)| M - M^3/3 + N <0$.
    \item On $V =-M$, $|W| \leq N$ we require $\dot{V} \geq -\sup_{s\in[0,T]} |f(s)| M + M^3/3 - N >0$.
    \item On $|V|\leq M$, $W= N$ we require $\dot{W} \leq \varepsilon (M - \gamma N + \beta ) < 0$.
    \item On $|V|\leq M$, $W= - N$ we require $\dot{W} \geq \varepsilon (-M + \gamma N + \beta ) > 0$.   
\end{itemize}
These requirements boil down to the following two conditions for the pair $(M, N)$
\[
M^3 - \sup_{s\in[0,T]} |f(s)| M  > N, \quad  \gamma N > M + |\beta|.
\]
If we denote $N = \theta M$, these conditions can be combined as 
\[
\frac{|\beta|}{\gamma M} + \frac{1}{\gamma} < \theta < M^2 - \sup_{s\in[0,T]} |f(s)|.
\]
It is easy to see that for any $\theta > 1/\gamma$, these inequalities are satisfied for all $M>0$ large enough. This means that given any $\theta > 1/\gamma$ the rectangles $R_k = [-M_k,M_k] \times [-N_k,N_k]$ with $M_k = k$, $N_k = \theta k$, $k \geq k_0$ are a family of nested positively invariant rectangles $R_k$ for system \eqref{positively_invariant_rectangles} with $\cup_{k\geq k_0} R_k = \R^2$.
In the particular case of $f(t)=1-S(\lambda t)^2\rho^2/2$, this implies (i) since the initial condition belongs to one of the positively invariant rectangles.

To establish (ii), since $f$, $V$, and $W$ are bounded, then  \eqref{positively_invariant_rectangles} implies that $\dot{V}$ is bounded as well.

Item (iii) is more delicate. Under the assumptions of Proposition \ref{prop_slope_freq}, we can compare the solutions $(V,W)$ of \eqref{PAS_slope_freq} with the solutions of the algebraic system \eqref{PAS_slope_freq_instantaneous}. Lemma \ref{properties_algebraic_system_HF} tells us that $X_1(t)^2 >1  - (\lambda t \rho)^2/2$, which together with the continuity and the fact that $X_1(t)$ is actually constant for $t \geq \frac{1}{\lambda}$ imply that
\[
X_1(t)^2 -1 + (\lambda t \rho)^2/2\geq c,
\]
for some constant $c>0$. Therefore, by choosing $0<\mu<\hat\mu$ and $0<\lambda<\hat \lambda$, Proposition \ref{prop_slope_freq} guarantees that 
\[ 
|V(t) - X_1(t)|^2 + |W(t) - Y_1(t)|^2 \leq \bar{C} \mu^2 + \bar{C}_1 \lambda,
\]
for some positive constants $\bar{C}$ and $\bar{C}_1$. And in this way, by choosing $\lambda^*\leq \hat{\lambda}$ and $\mu^* \leq \hat{\mu}$ so that
\[
\bar{C} (\mu^*)^2 + \bar{C}_1 \lambda^* \leq c/2,
\]
we obtain $V(t)^2 -1 + (\lambda t \rho)^2/2\geq  c/2 >0$ for all $0< \lambda < \lambda^*$ and $0<\mu < \mu^*$, which gives us (iii). This concludes the proof of Proposition \ref{prop_properties_pas_HF}.
\end{proof}

\begin{proposition}[Validity of approximation result for system \eqref{PAS_slope_freq}]\label{proposition_error_estimates_slope_HF}
Suppose that the parameters ($\varepsilon$, $\beta$,  $\gamma$, $\rho$) and $M, \omega_0$ satisfy Condition \ref{conditionC1}  in Definition \ref{definition_condition_parameters} and  Condition \ref{conditionC3b} in Proposition \ref{proposition_parameters}. Let $\lambda^*, \mu^*>0$ given by Proposition \ref{prop_properties_pas_HF}. Then for all  $0<\lambda \leq \min\{\lambda^*,M\}$ and $T= 1/\lambda$, if the corresponding solution $(V,W)$ of \eqref{PAS_slope_freq} satisfy \eqref{smallness_initial_data_PAS} then the functions $\varphi_\textit{slow}$, $\varphi_\textit{fast}$, $\varphi_2$, $H\!F_1$, $H\!F_2$ given in Proposition \ref{defi_approximation_problem_HF} satisfy the hypotheses of Lemma \ref{lemma_linear_error_abs}. Therefore, condition \ref{conditionC3} is satisfied.
\end{proposition}

\begin{proof}

We have to verify each condition of Lemma \ref{lemma_linear_error_abs}. Let $0\leq t \leq T = 1/\lambda$ and let $\omega_0$ be given by Condition \ref{conditionC3b}.
\begin{itemize}
\item {\bf Condition \ref{condition_lemma_1} :} First, we estimate the integral of $\varphi_\textit{fast}$ using Lemma \ref{integral_estimate_slow_varying}
\begin{align*}
\int_\tau^t\varphi_\textit{fast}(s) ds &= \int_\tau^t \left(\frac{(\rho \lambda s)^2}{2} \cos(2\omega s) -2V(s) (\rho \lambda s) \sin(\omega s)\right) ds\\
&\leq \frac{1}{2\omega} \min\{(t-\tau) 
 2\omega,2\} \frac{\rho^2\lambda^2 t^2}{2 } + \frac{t-\tau}{2\omega} (\rho^2\lambda^2 t) \\
 &\hspace{1cm}+ \frac{1}{\omega} \min  \{(t-\tau) \omega ,2\} 2 (|\rho| \lambda t)\max_{s\in [\tau ,t]}|V(s)|\\
 & \hspace{1cm}+ \frac{t-\tau}{\omega} 2 |\rho| \lambda \max_{s\in [\tau ,t]}|V(s)| + \frac{t-\tau}{\omega} 2 (|\rho| \lambda t)\max_{s\in [\tau ,t]}|\dot{V}(s)|.
\end{align*}
Second, we bound $\varphi_\textit{slow}$ using the uniform bound for $|V(s)|$ provided by Proposition \ref{prop_properties_pas_HF}
\[
\max_{s\in[0,T]}\varphi_\textit{slow}(s) \leq -\min_{s\in[0,T]}\left(|V(s)|^2 -1 + \frac{(\rho \lambda s)^2}{2} \right)  =: -K_\textit{slow} <0.
\]
Combining those estimates, we obtain 
\begin{align}\notag
E &= \int_\tau^t(\varphi_\textit{slow}(s) + \varphi_\textit{fast}(s))ds \\
&\leq  -K_\textit{slow}(t-\tau) \notag\\
&\hspace{1cm}+ \left(\frac{\rho^2 \lambda^2 t}{2\omega} + \frac{2|\rho| \lambda}{\omega} \max_{s\in[\tau,t]}|V(s)| + \frac{2 |\rho| \lambda t}{\omega} \max_{s\in[\tau,t]}|\dot V(s)| \right) (t-\tau) \notag \\
& \hspace{1cm}+ \left(\frac{\rho^2 \lambda^2 t^2}{2} + 4 |\rho| \lambda t \max_{s\in[\tau,t]}|V(s)|\right)\frac{1}{\omega} \notag \\ \label{integral_varphis}
&\leq \left(-K_\textit{slow} + \frac{A}{\omega}\right)(t-\tau) + \frac{B}{\omega}
\end{align}
where 
\begin{equation*}
A = \frac{\rho^2 \lambda}{2} + 2|\rho| \lambda \max_{s\in[0,T]}|V(s)| + 2 |\rho| \max_{s\in[0,T]}|\dot{V}(s)|, \quad B = \frac{\rho^2}{2} + 4 |\rho| \max_{s\in[0,T]}|V(s)|.
\end{equation*}
Integrating, we get
\begin{align*}
\frac{1}{\gamma} \int_0^t e^{\int_\tau^t(\varphi_\textit{slow}(s) + \varphi_\textit{fast}(s))ds}d\tau &\leq \frac{1}{\gamma}\int_0^t e^{(-K_\textit{slow} + A/\omega)(t-\tau) + B/\omega} d\tau\\
&= \frac{e^{B/\omega}}{\gamma} \frac{1-e^{(-K_\textit{slow} + A/\omega)t}}{K_\textit{slow}-A/\omega}.
\end{align*}
Therefore, Condition \ref{conditionC3b} implies   \ref{condition_lemma_1}.

\item {\bf Condition \ref{condition_lemma_2}:} Let $t\in [0,1/\lambda]$ then
\begin{align*}
I(t) &= \int_0^t e^{\int_\tau^t\left( \varphi_\textit{slow} + \varphi_\textit{fast} \right) ds}H\!F_1(\tau) d\tau\\
&=\int_0^t e^{\int_\tau^t\left( \varphi_\textit{slow} + \varphi_\textit{fast} \right) ds}\Big( (- V(\tau)^2) (\rho \lambda \tau) \sin(\omega \tau) + \frac{V(\tau) (\rho \lambda \tau)^2}{2} \cos(2\omega \tau) \\
& \hspace{2cm}-\frac{1}{3}(\rho \lambda \tau)^3 \sin^3(\omega \tau)+ (\rho \lambda \tau)\sin(\omega \tau) - \rho \lambda \sin(\omega \tau)\Big) d\tau.
\end{align*}
We use the identity $\sin^3 (x)= \frac{3}{4}\sin (x) -\frac{1}{4}\sin(3x)$ and Lemma \ref{integral_estimate_slow_varying} to obtain
\begin{align*}
|I(t)| & \leq  \frac{1}{\omega}\left(2 + K \max_{\tau\in [0 ,t]} \left|e^{\int_\tau^t\left( \varphi_\textit{slow} + \varphi_\textit{fast} \right) ds}\right|\right) \\
&\hspace{1cm}\times\Big( \max_{\tau\in [0 ,t]} \left|e^{\int_\tau^t\left( \varphi_\textit{slow} + \varphi_\textit{fast} \right) ds}  V(\tau)^2 (\rho \lambda \tau) \right|\\
& \hspace{2cm}+ \frac{1}{4}\max_{\tau\in [0 ,t]} \left|e^{\int_\tau^t\left( \varphi_\textit{slow} + \varphi_\textit{fast} \right) ds} V(\tau) (\rho \lambda \tau)^2\right|\\
&\hspace{2cm} +  \frac{5}{18}\max_{\tau\in [0 ,t]} |e^{\int_\tau^t\left( \varphi_\textit{slow} + \varphi_\textit{fast} \right) ds} (\rho \lambda \tau)^3|\\
& \hspace{2cm} + \max_{\tau\in [0 ,t]} |e^{\int_\tau^t\left( \varphi_\textit{slow} + \varphi_\textit{fast} \right) ds} (\rho \lambda \tau)|\\
&\hspace{2cm}
 + \max_{\tau\in [0 ,t]} |e^{\int_\tau^t\left( \varphi_\textit{slow} + \varphi_\textit{fast} \right) ds}\rho \lambda| \Big)\\
& + \frac{t}{\omega} \Big( \max_{\tau\in [0 ,t]} \left|e^{\int_\tau^t\left( \varphi_\textit{slow} + \varphi_\textit{fast} \right) ds}  (2 |V(\tau)| |\dot{V}(\tau)| (\rho \lambda \tau)  + |V(\tau)|^2 \rho \lambda )\right|\\
& \hspace{1cm} + \frac{1}{4}\max_{\tau\in [0 ,t]} \left|e^{\int_\tau^t\left( \varphi_\textit{slow} + \varphi_\textit{fast} \right) ds} (|\dot{V}(\tau)| (\rho \lambda \tau)^2 + 2 |V(\tau)| \rho^2 \lambda^2 \tau ) \right|\\
& \hspace{1cm}+  \frac{5}{18} \max_{\tau\in [0 ,t]} |e^{\int_\tau^t\left( \varphi_\textit{slow} + \varphi_\textit{fast} \right) ds} 3 \rho^3 \lambda^3 \tau^2|\\
&\hspace{1cm}
 + \max_{\tau\in [0 ,t]} |e^{\int_\tau^t\left( \varphi_\textit{slow} + \varphi_\textit{fast} \right) ds} \rho \lambda | \Big)\\
&\leq \frac{(1+t)C}{\omega},
\end{align*}
where $K=\max_{\tau\in[0,t]}|\varphi_\textit{slow} + \varphi_\textit{fast}|\leq 1+(\max_{\tau\in[0,t]}|V(\tau)|+|\rho|)^2$. In the previous estimate, $C$ is 
uniformly bounded independent of $t$ and $\omega\geq \omega_0$, because of Proposition \ref{prop_properties_pas_HF}, the observation that $\lambda t \leq 1$ for all $t \in [0,1/\lambda]$, and because equation \eqref{integral_varphis} and Condition \ref{conditionC3b} imply $e^{\int_\tau^t\varphi_\textit{slow} + \varphi_\textit{fast} ds}\leq e^{B/\omega}$ is uniformly bounded independent of $t$ and $\omega\geq \omega_0$.

\item {\bf Condition \ref{condition_lemma_3}:} Directly from Lemma \ref{integral_estimate_slow_varying} we obtain
\begin{align*}
    \sup\limits_{t\in [0, T]}\left|\int_0^t  e^{-\varepsilon\gamma (t-\tau)} H\!F_2(\tau) d\tau\right| &\leq \sup\limits_{t\in [0, T]}\left|\int_0^t  e^{-\varepsilon\gamma (t-\tau)} \varepsilon (\rho \lambda \tau) \sin(\omega \tau) d\tau\right|\\ 
    &\leq \frac{(1+T)C}{\omega}.
\end{align*}

\item {\bf Condition \ref{condition_lemma_4}:} The bound is $\varphi_2(t)$ is a direct consequence of the bound on $V$ provided by Proposition \ref{prop_properties_pas_HF}
\[
|\varphi_2(s)| \leq |\rho| (\lambda t) + \max_{s\in [0\,T]}|V(s)| \leq |\rho| + \max_{s\in [0\,T]}|V(s)| =: K_2.
\]
\end{itemize}
Since we just verified that conditions of Lemma \ref{lemma_linear_error_abs} are satisfied for $0 < \lambda \leq M^* = \min \{\lambda^*,M\}$, $T = 1/\lambda$, and $\omega\geq \omega_0$, for each fixed value of $\lambda\in [0,M^*]$ we can take the limit $\omega\to \infty$, uniformly in $t\in[0,T]$, in the estimate provided by Lemma \ref{lemma_linear_error_abs}, hence obtaining Condition \ref{conditionC3}. This concludes the proof of Proposition \ref{proposition_error_estimates_slope_HF}. 
\end{proof}

\subsection{Estimates for the approximation error, steering problem 2}\label{subsec_approximation_dc_term}

\begin{proposition}[Equation for the approximation error: steering problem 2]\label{defi_approximation_problem_DC}
Let $(v,w)$ be the solution of \eqref{FHN_hf_intro} with input current \eqref{input_current_slope_general2} and initial data $(v_0,w_0)$. Let $(V,W)$ be the solution of the partially averaged system \eqref{FHN_pas} (which has the same initial data $(v_0,w_0)$). Consider the approximation error given by $E_v = v -V - \rho \sin(\omega t)$, $E_w = w - W$, then $(E_v, E_w)$ satisfy 
\begin{equation}\label{error_equation}
\left\{
\begin{array}{rll}
\dot{E}_v  &=  \left(\varphi_\textit{slow} + \varphi_\textit{fast}\right)E_v + \varphi_2 E_v^2  -\frac{1}{3}E_v^3 -E_w+ H\!F_1 &, 0<t<\frac{1}{\delta},\\
\dot{E}_w &= \varepsilon( E_v - \gamma E_w) + H\!F_2&, 0<t<\frac{1}{\delta},\\
E_v(0)&= 0,\quad E_w(0) = 0,
\end{array}\right.
\end{equation}
where 
\begin{align}
\varphi_\textit{slow} &= 1-V^2 - \rho^2/2, \label{defi_varphi11}\\
\varphi_\textit{fast} &= \frac{\rho^2}{2} \cos(2\omega t) -2V \rho \sin(\omega t),\label{defi_varphi12}\\
\varphi_2 &= -(\rho \sin(\omega t) + V), \label{defi_varphi2}\\
H\!F_1 &=- V^2 \rho \sin(\omega t) + \frac{V \rho^2}{2} \cos(2\omega t)   -\frac{1}{3}\rho^3 \sin^3(\omega t)+ \rho\sin(\omega t), \label{defi_varphi3}\\
H\!F_2 &= \varepsilon \rho \sin(\omega t)\label{defi_HF2}.
\end{align}

\end{proposition}

The following proposition gives us the tool necessary to prove the approximation result.

\begin{proposition}[Properties of system \eqref{FHN_pas}]\label{prop_properties_pas_DC} 
Let $(\varepsilon,\beta, \gamma,  \rho, I_0)$ satisfy Condition \ref{conditionC2} in Definition \ref{definition_condition_parameters} and let $\hat{\mu}, \hat{\delta} >0$ given by Proposition \ref{lemma_slow_steer}. 
There exists $\delta^* \leq \hat{\delta}$ and $\mu^*\leq \hat{\mu}$ such that for all $0<\delta \leq  \delta^*$ the corresponding solution $(V,W)$ of system \eqref{FHN_pas} with initial conditions satisfying 
\begin{equation}\label{smallness_initial_data_PAS2}
\|(V(0) - v_1, W(0)-w_1 )\|_2 \leq \mu^*,
\end{equation}
where $(v_1,w_1)$ is given by \eqref{definition_P1},  satisfies the following for some constants $C_1, C_2 >0$:
\begin{enumerate}[label=(\roman*)]
\item $\max\limits_{s\in[0, 1/\delta]} |V(s)| \leq C_1$,
\item $\max\limits_{s\in[0, 1/\delta]} |\dot{V}(s)| \leq C_4$,
\item $V(s)^2 >1 - \rho^2/2$.
\end{enumerate}
\end{proposition}

The proof of Proposition \ref{prop_properties_pas_DC} is analogous to the proof of Proposition \ref{prop_properties_pas_HF}.

\begin{proposition}\label{proposition_error_estimates_slope_DC}
Suppose that the parameters ($\varepsilon$, $\beta$,  $\gamma$, $\rho$, $I_0$) and let $M, \omega_0>0$ satisfy Condition \ref{conditionC2} in Definition \ref{definition_condition_parameters} and Condition \ref{conditionC4b} in Proposition \ref{proposition_parameters}. Let $\mu^*,\delta^* >0$ be given by Proposition \ref{prop_properties_pas_DC}. Then for all $0<\delta \leq \min\{\delta^*, M\}$ and $T= 1/\delta$, if the corresponding solution $(V,W)$ of \eqref{FHN_pas} satisfy \eqref{smallness_initial_data_PAS2} then the functions $\varphi_\textit{slow}$, $\varphi_\textit{fast}$, $\varphi_2$, $H\!F_1$, $H\!F_2$ given by 
Proposition \ref{defi_approximation_problem_DC}
satisfy the hypothesis of Lemma \ref{lemma_linear_error_abs}. Therefore, Condition \ref{conditionC4} is satisfied.
\end{proposition}

The proof of this proposition is analogous to the proof of Proposition \ref{proposition_error_estimates_slope_HF} (with $\lambda t =1$ and replacing Proposition \ref{prop_properties_pas_HF} with Proposition \ref{prop_properties_pas_DC}.

\subsection{Proof of Theorem \ref{thm_slow_steering1}}\label{subsec_proof_thm1}

Let ($\varepsilon$, $\gamma$, $\beta$, $\rho$) satisfy Condition \ref{conditionC1} and Condition \ref{conditionC3} in Definition \ref{definition_condition_parameters}. Let $(v_0,w_0)\in \R^2$ be given by the solution of \eqref{definition_P0}, and let $\eta>0$. For each $\lambda >0$ let $(V,W)$ be the corresponding solution of the partially averaged system \eqref{PAS_slope_freq} (with initial data $(v(0),w(0))$, and let $(X_1,Y_1)$ be the corresponding solution of \eqref{PAS_slope_freq_instantaneous}. 

Because of Condition \ref{conditionC1} system \eqref{PAS_slope_freq} is autonomous and stable around $(v_1, w_1)$ given by \eqref{definition_P1} for $t > 1/\lambda$. The averaging theorem \cite[Theorem 10.4]{khalil2013} tells us that given $\eta>0$ there exist $\hat{\omega}>0$ and $\eta_1 > 0 $ such that for all $\omega\geq \hat{\omega}$ if 
\begin{equation}\label{condition_approximation_averaging}
\left\|\left(v(1/\lambda ) - \rho\sin(\omega /\lambda ) - V(1/\lambda ), w(1/\lambda ) - W(1/\lambda ) \right)\right\|_2 \leq \eta_1,
\end{equation}
then we can guarantee that for all time $t > 1/\lambda$
\[
\left\|\left(v(t ) - \rho\sin(\omega t) - V(1/\lambda ), w(t) - W(1/\lambda) \right)\right\|_2 \leq \eta.
\]

Second, for time $0\leq t \leq 1/\lambda$, by Proposition \ref{prop_slope_freq} there exists $ \lambda_1, \hat{\mu}, C, C_1>0$, so that if $\lambda \leq \lambda_1$ and the initial condition satisfy
\[
\mu =  \|(V(0) - v_0,W(0) - w_0)\|_2 \leq \hat{\mu},
\]
then $\sup_{t\in[0,1/\lambda]}\|(V(t) - X_1(t), W(t) - Y_1(t)\|_2^2  \leq C \mu^2 + C_1 \lambda$. We obtain that by choosing $\mu^* \leq \hat{\mu}$ and $\lambda_2 \leq \lambda_1$ so that $C (\mu^*)^2 + C_1 \lambda_2\leq (\eta_1/2)^{2}$, we can guarantee that if $\lambda \leq \lambda_2$, and $\mu \leq \mu^*$, then we have 
\[
\|(V(t) - X_1(t), W(t) - Y_1(t))\|_2  \leq \eta_1/2, \quad  \forall t\in [0,1/\lambda]. 
\]
Third, because of Condition \ref{conditionC3} we know there exits  $M > 0$ such that for each $0<\lambda \leq M$ there exists $\omega_0 = \omega_0(\lambda) \geq\hat{\omega}$ so that for all $\omega \geq \omega_0$ 
\[
\sup_{t\in[0,1/\lambda]}\|(E_v(t), E_w(t))\|_2\leq \eta_1/2,
\]
where $E_v = v- V- \rho \lambda t \sin(\omega t)$, $E_w = w- W$.

Combining our estimates, we obtain that for $ \lambda \leq \lambda^*=\min\{\lambda_2,M\}$, $\mu \leq \mu^*$, $\omega \geq \omega_0(\lambda)$, and all $t \in [0,1/\lambda]$
\begin{align*}
I &= \|\left(v(t)-\rho\lambda t \sin(\omega t) - X(t), w - Y(t)\right)\|_2\\
&\leq \|(v(t)-\rho \lambda t \sin(\omega t) - V(t), w(t) - W(t))\|_2\\
& \hspace{1cm}+ \|\left(V(t) - X_1(t), W(t) - Y_1(t)\right)\|_2\\
&= \|(E_v(t), E_w(t))\|_2 + \|(V(t) - X_1(t), W(t) - Y_1(t))\|_2 \\
&\leq \eta_1/2 +\eta_1/2= \eta_1.
\end{align*}
This means that \eqref{condition_approximation_averaging} is satisfied, and therefore the estimate is valid for all $t>0$. This concludes the proof of Theorem \ref{thm_slow_steering1}.\qed

\subsection{Proof of Theorem \ref{thm_slow_steering2}}\label{subsec_proof_thm2}
Let ($\varepsilon$, $\gamma$, $\beta$, $\rho$,  $I_0$) satisfying Condition \ref{conditionC2} and Condition \ref{conditionC4} in Definition \ref{definition_condition_parameters}. Let $(v_1,w_1)\in \R^2$ be given by the solution of \eqref{definition_P1}, and let $\eta >0$. For each $\delta>0$ Let $(V,W)$ be the solution of the corresponding partially averaged system \eqref{FHN_pas} (with initial $(v(0),w(0))$ and $(X_2,Y_2)$ be the corresponding solution of \eqref{instantaneous_system}. We follow a similar approach to the proof of Theorem \ref{thm_slow_steering1}.

First, for $t > 1/\delta$ thanks to the averaging theorem, and because of Condition \ref{conditionC2} system \eqref{FHN_pas} is autonomous and stable around $(v_2, w_2)$ given by \eqref{definition_P2} for $t > 1/\delta$, we know that given $\eta>0$ there exists $\hat{\omega}>0$ and $\eta_1 > 0 $ such that for all $\omega\geq \hat{\omega}$ if 
\[
\|(v(1/\delta) - \rho \sin(\omega /\delta) - X_2(1/\delta), w(1/\delta)  - Y_2(1/\delta))\|_2 \leq \eta_1,
\]
then for all time $t > 1/\delta$ we have
\begin{equation}\label{condition_approximation_averaging2}
\|(v(t) - \rho \sin(\omega t) - X_2(1/\delta), w(t)  - Y_2(1/\delta))\|_2 \leq \eta.
\end{equation}
Second, for $0<t<1/\delta$ by Proposition \ref{lemma_slow_steer} there exits $\delta_1, \hat{\mu}, C, C_1>0$ such that if $0<\delta< \delta_1$, and the initial condition satisfies
\[
\mu = \left\|(V(0) - v_1,W(0) - w_1)\right\|_2  \leq \hat{\mu}, 
\]
then $\sup_{t\in[0,1/\delta]} \left\|(V(t) - X_1(t),W(t) - Y_1(t))\right\|_2^2 \leq C \mu^2 + C_1 \delta$. We obtain that by choosing $\mu^*\leq \hat{\mu}$ and $\delta_2\leq\delta_1$ such that $C (\mu^*)^2 + C_1 \delta_2\leq (\eta_1/2)^{2}$ then we can guarantee that if $\delta \leq \delta_2$, and $\mu \leq \mu^*$, then we have 
\[
\|(V(t) - X_2(t), W(t) - Y_2(t)\|_2  \leq \eta_1/2, \quad  \forall t\in [0,1/\delta]. 
\]
Third, because of Condition \ref{conditionC4} we know there exits  $M>0$ such that for each $0<\delta \leq M$, there exists $\omega_0 = \omega_0(\delta) \geq\hat{\omega}$ so that for all $\omega \geq \omega_0$ 
\[
\sup_{t\in[0,1/\delta]}\|(E_v(t), E_w(t))\|_2\leq \eta_1/2,
\]
where $E_v = v- V- \rho \sin(\omega t)$, $E_w = w- W$. 
Combining our estimates, we obtain that for $ \delta \leq \delta^*=\min\{\delta_2,M\}$, $\mu \leq \mu^*$, $\omega \geq \omega_0(\delta)$, and all $t \in [0,1/\delta]$,
\begin{align*}
I &= \|\left(v(t)-\rho\sin(\omega t) - X_2(t), w - Y_2(t)\right)\|_2 \\&\leq \|(v(t)-\rho\sin(\omega t) - V(t), w(t) - W(t))\|_2\\
& \hspace{1cm} + \|\left(V(t) - X_2(t), W(t) - Y_2(t)\right)\|_2\\
&\leq \|(E_v, E_w)\|_2  + \|(V(t) - X_2(t), W(t) - Y_2(t))\|_2 \\
&\leq \eta_1/2 + \eta_1/2 = \eta_1.
\end{align*}
This implies that \eqref{condition_approximation_averaging2} is satisfied, and therefore, the estimate is valid for all $t>0$. This concludes the proof of Theorem \ref{thm_slow_steering2}.\qed

\section{Numerical Experiments}\label{section_numerical_simulations}

We performed two numerical experiments to illustrate the dependence on the parameters for the slope given by  Theorem \ref{thm_slow_steering1} and Theorem \ref{thm_slow_steering2}. 
All simulations were carried out using the Myokit package (v1.35.4) \cite{Clerx2016Myokit} for Python (v3.11.5), and all source code is available at
\url{https://github.com/estebanpaduro/qs-simulations}. For both experiments, unless stated otherwise, we used $\varepsilon=0.08$, $\gamma=0.5$, $\beta = 0.8$. In addition, note that what we termed the ``amplitude'' parameter $\rho$ in the averaged system corresponds to an oscillatory source of frequency $\omega$ and amplitude $\rho \omega$ for $\omega\gg 1$ as in Equation \eqref{current_no_slope}. Action potentials were detected using two conditions: 1) $V(t)$ exceeded a threshold value, $v^* = 1$, and 2) the prominence, \textit{i.e.}, the difference between baseline and maximum peak value, was greater than 1.

\subsection{Experiment 1: Slope of the envelope of the HFBS} 
Consider the averaged FHN system \eqref{PAS_slope_freq}, which corresponds to the partial averaging of system \eqref{FHN_hf_intro} with a HFBS \eqref{input_current_slope_general1}. On each simulation, we set a value of $\beta$ and determined if action potentials were generated for different values of the amplitude of the HFBS, $\rho$, and the slope of the envelope of the HFBS, $\lambda$. In all cases, the initial conditions were those of the equilibrium \eqref{definition_P0}, and there was no DC component, \textit{i.e.}, $I_0=0$. Figure \ref{fig:FHN_HF} summarizes the results of Experiment 1. For $\beta = 0.75$, there was a well-defined region in the ($\lambda$, $\rho$) plane for which an onset action potential was observed for $0<t<200$ (Figure \ref{figure2a}). In fact, for sufficiently low amplitude ($\rho <\sim 0.44$), no action potentials were elicited regardless of the slope of the HFBS. However, for higher amplitudes, onset activation could be avoided using a sufficiently gradual slope. For instance, for $\rho = 0.6$, onset response was observed for $\lambda = 0.9$, but not for $\lambda = 0.04$ (Figure \ref{figure2b}). Moreover, this effect was dependent on the adaptation parameter $\beta$. Figure \ref{figure2c} shows the interface between the region with onset action potentials and the region without action potentials, as observed in Figure \ref{figure2a}, for different values of $\beta$. Each boundary curve was obtained by searching the smallest value of $\rho$ such that an onset action potential was observed for each value of $\lambda$. We chose a range of values for $\beta$ where condition \ref{conditionC1} is satisfied, \textit{i.e.}, $\beta\in [0.65, 0.9]$ with a step size of $\Delta_\beta = 0.025$. This result indicates that, given the amplitude of the HFBS, it is possible to modulate its envelope to avoid onset action potentials. Our findings are consistent with prior theoretical and experimental work and suggest that a careful tuning of the HFBS may minimize onset responses in a manner strongly dependent on the properties of the neuron model. Consequently, conduction block applications should consider the electrophysiological characteristics of target neurons for optimal design of waveforms that avoid onset responses.

\begin{figure}[!ht]
\centering
\subfloat[]{\label{figure2a}\includegraphics[width=0.37\textwidth]{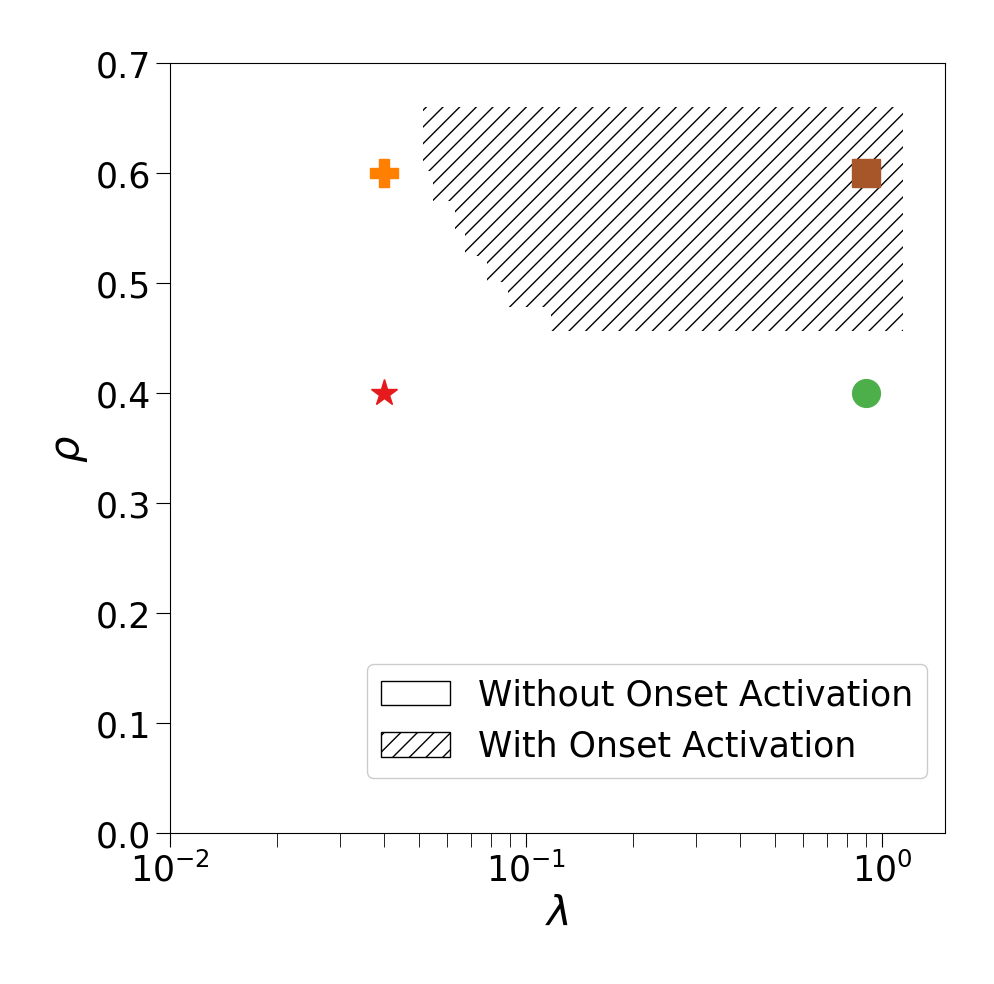}}\hfill
\subfloat[]{\label{figure2b}\includegraphics[width=0.62\textwidth]{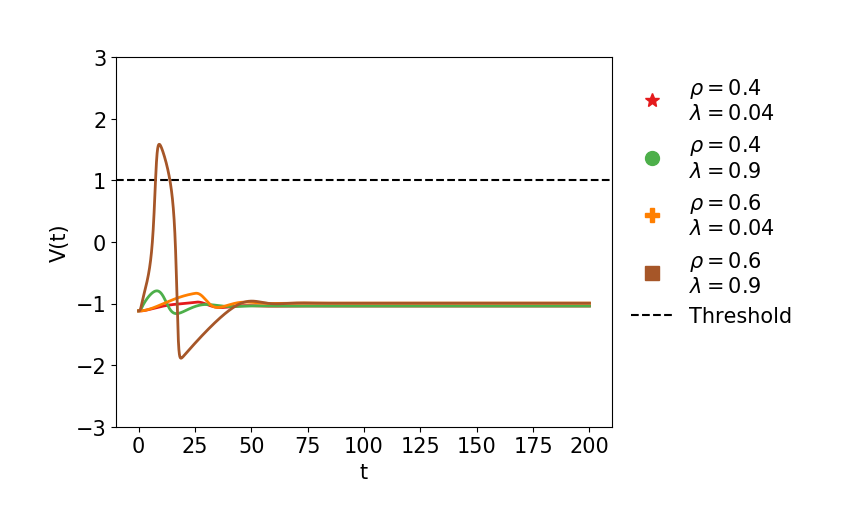}}\\
\subfloat[]{\label{figure2c}\includegraphics[width=0.5\textwidth]{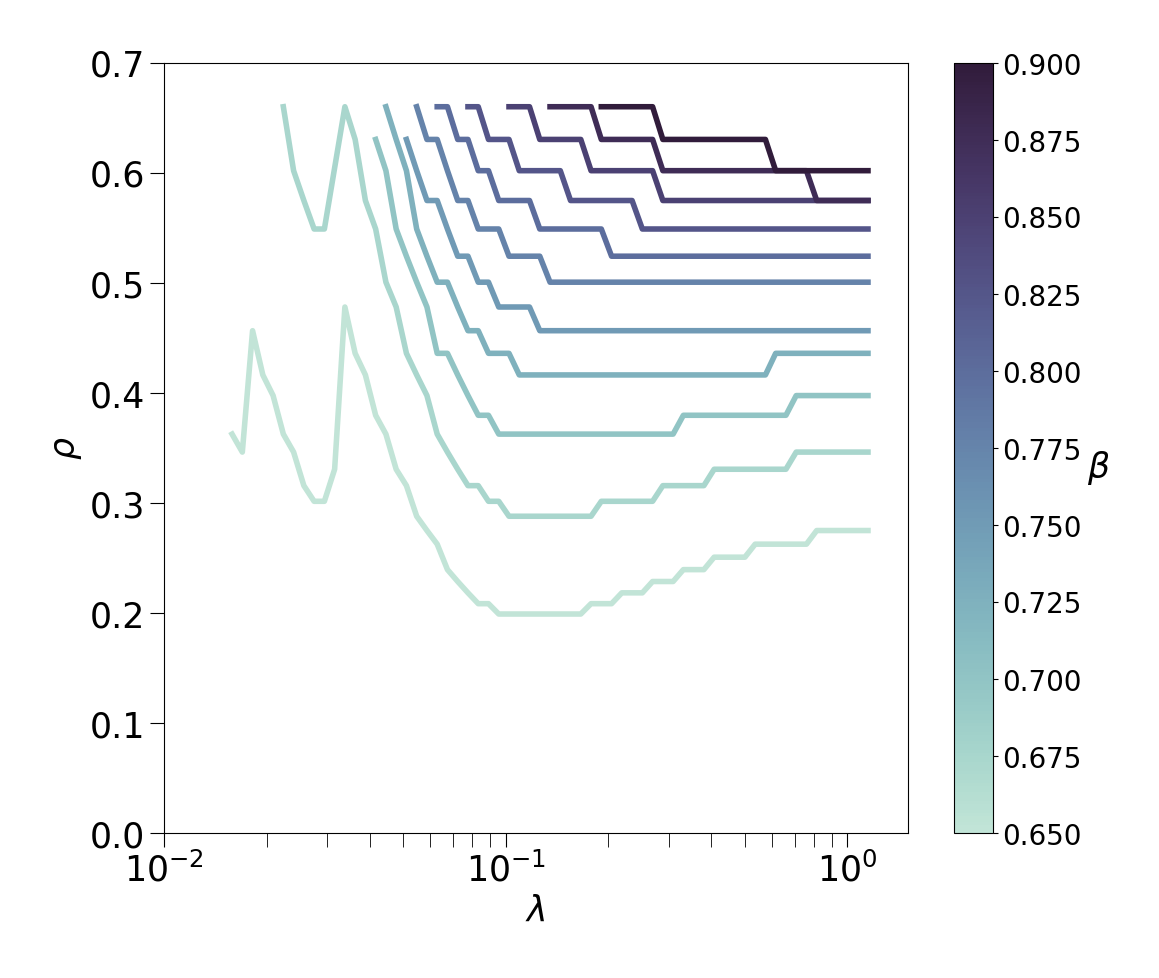}}
\caption{Experiment 1, Slope of the envelope of the HFBS. (a) Region of action potential generation for a given combination of $\rho$ and $\lambda$.  In this case, $\beta = 0.75$. (b) Examples of the recorded potential, $V(t)$, for different pairs of $\rho$ and $\lambda$ for the averaged system \eqref{PAS_slope_freq} as indicated in (a). (c) Boundary curve that separates the region with and without onset response, respectively, for different values of $\beta$. Note that in (a) and (c), the $\lambda$ axis is logarithmically spaced.
}\label{fig:FHN_HF}
\end{figure}

\subsection{Experiment 2: Incorporating the DC term with a ramp} Consider the FHN system \eqref{FHN_pas} corresponding to the partial averaging of system \eqref{FHN_hf_intro} with an input current of the form \eqref{input_current_slope_general2}. In this case, we ignored the transient effects of the HFBS input by letting $T_{w} = 100$, and then we incorporated a ramp in the DC term with slope $\delta$. On each simulation, we set the DC component to a value of $I_0$, which was reached after a ramp of slope $\delta$, and after initiating the HFBS of amplitude $\rho$. In all cases, the initial conditions were those of the equilibrium \eqref{definition_P1}, and we selected ranges of the parameters to observe regions with and without action potentials in $t\in[0,500]$. Figure \ref{fig:FHN_DC} summarizes the results of Experiment 2. Similar to Experiment 1, for a small DC component ($I_0=0.2$), we observed well-defined regions with and without onset response, and the onset response exhibited a single action potential in this case (Figure \ref{figure3a}). Further, depending on the value of $\rho$, the stationary value attained in $V(t)$ may change (Equation \eqref{definition_P2}), but only in specific cases noticeable (for example, $(\rho,\delta) = (0.5,0.3)$) we observed an onset action potential because the slope parameter $\delta$ was too large relative to $\rho$ (Figure \ref{figure3b}). In addition, for a larger DC component ($I_0 = 0.4$), we observed three different regions with distinct qualitative behavior: one without onset action potentials, a second one for which a single onset action potential was elicited, and third a region for which $\rho$ was too small with respect to the current $I_0$ and therefore persistent excitation was observed (Figure \ref{figure3c} and \ref{figure3d}). In the latter case, we note that no value of $\delta$ can avoid the generation of action potentials. The shape and location of these regions were strongly dependent on the amplitude of the DC component (Figure \ref{figure3e}). Similar to Experiment 1, for each value of $I_0$, we defined a boundary curve for which, given a value of $\delta$, we seek the largest value of $\rho$ such that at least one action potential was elicited. These findings are consistent with Theorem \ref{thm_slow_steering2} since for pairs $(I_0, \rho)$ to which no oscillatory behavior is observed, we can always find a suitable small value for the slope parameter $\delta$ so that all onset action potentials can be avoided. 

\begin{figure}[!htp]
\centering
\subfloat[]{\label{figure3a}\includegraphics[width=0.35\textwidth]{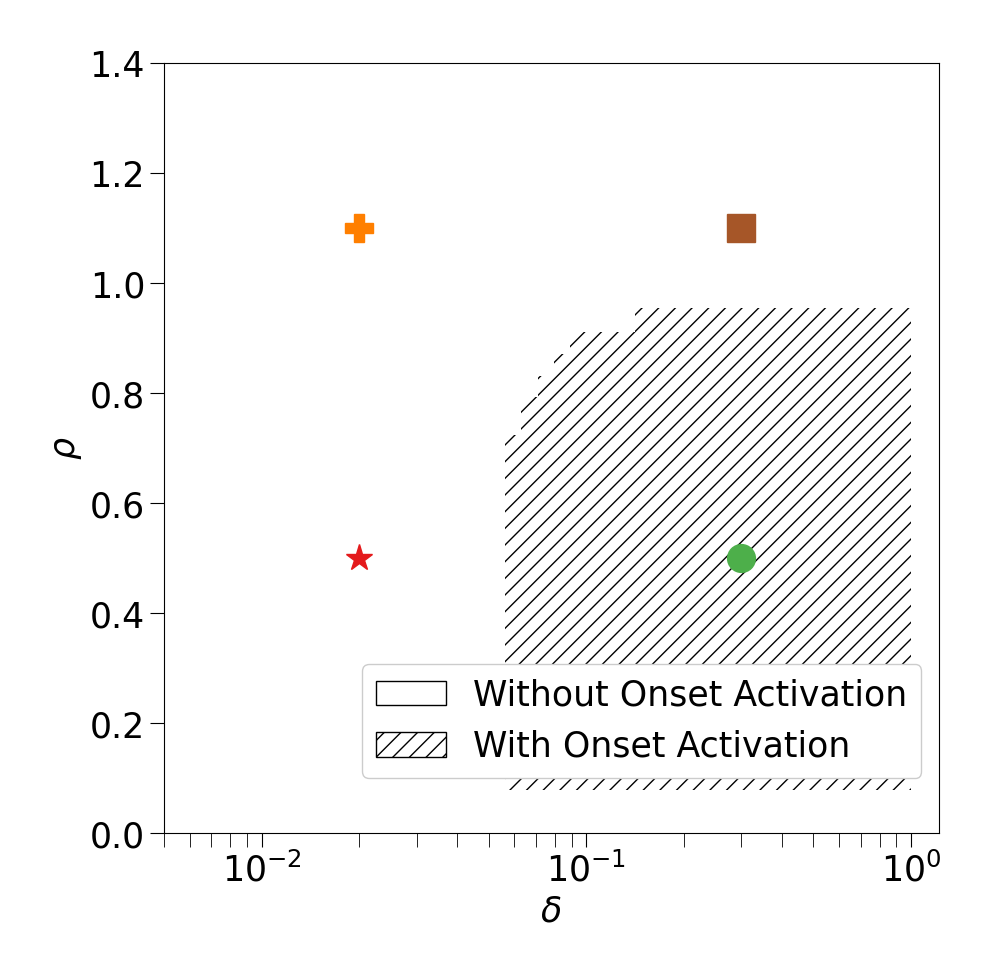}}
\subfloat[]{\label{figure3b}\includegraphics[width=0.55\textwidth]{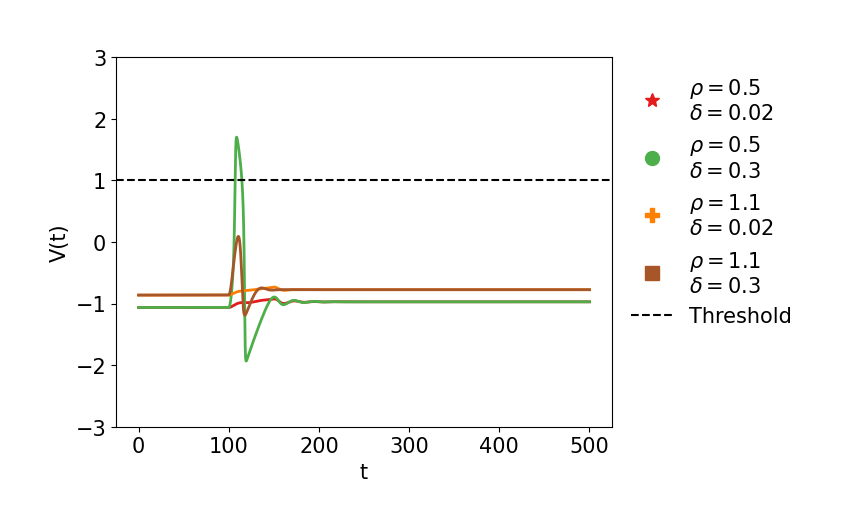}}\\
\subfloat[]{\label{figure3c}\includegraphics[width=0.35\textwidth]{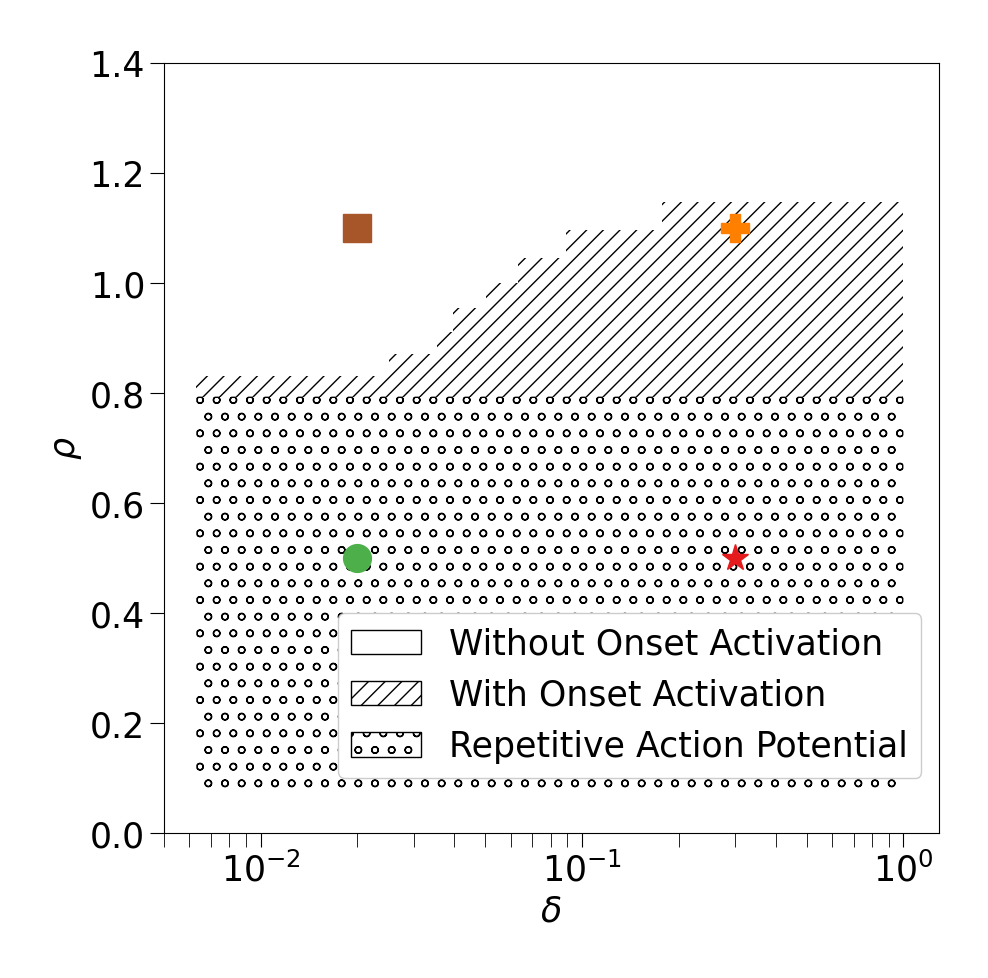}}
\subfloat[]{\label{figure3d}\includegraphics[width=0.55\textwidth]{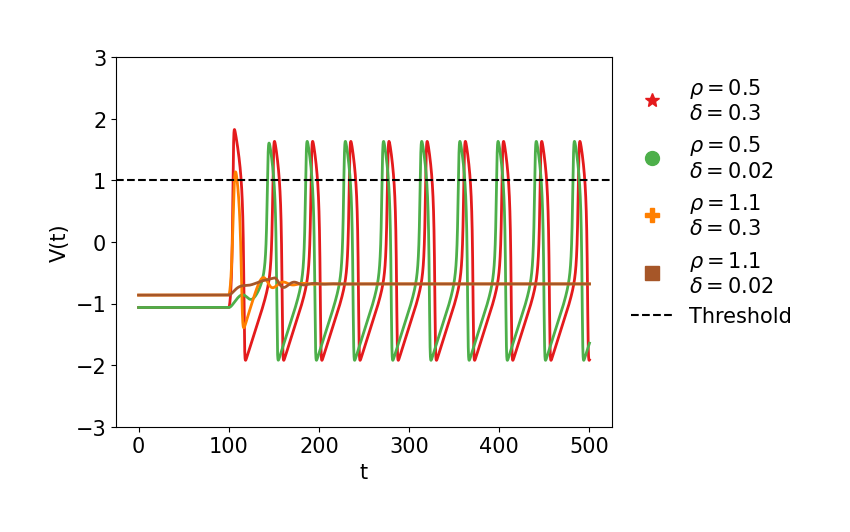}}\\
\centering
\vspace{-2mm}
\subfloat[]{\label{figure3e}\includegraphics[width=0.39\textwidth]{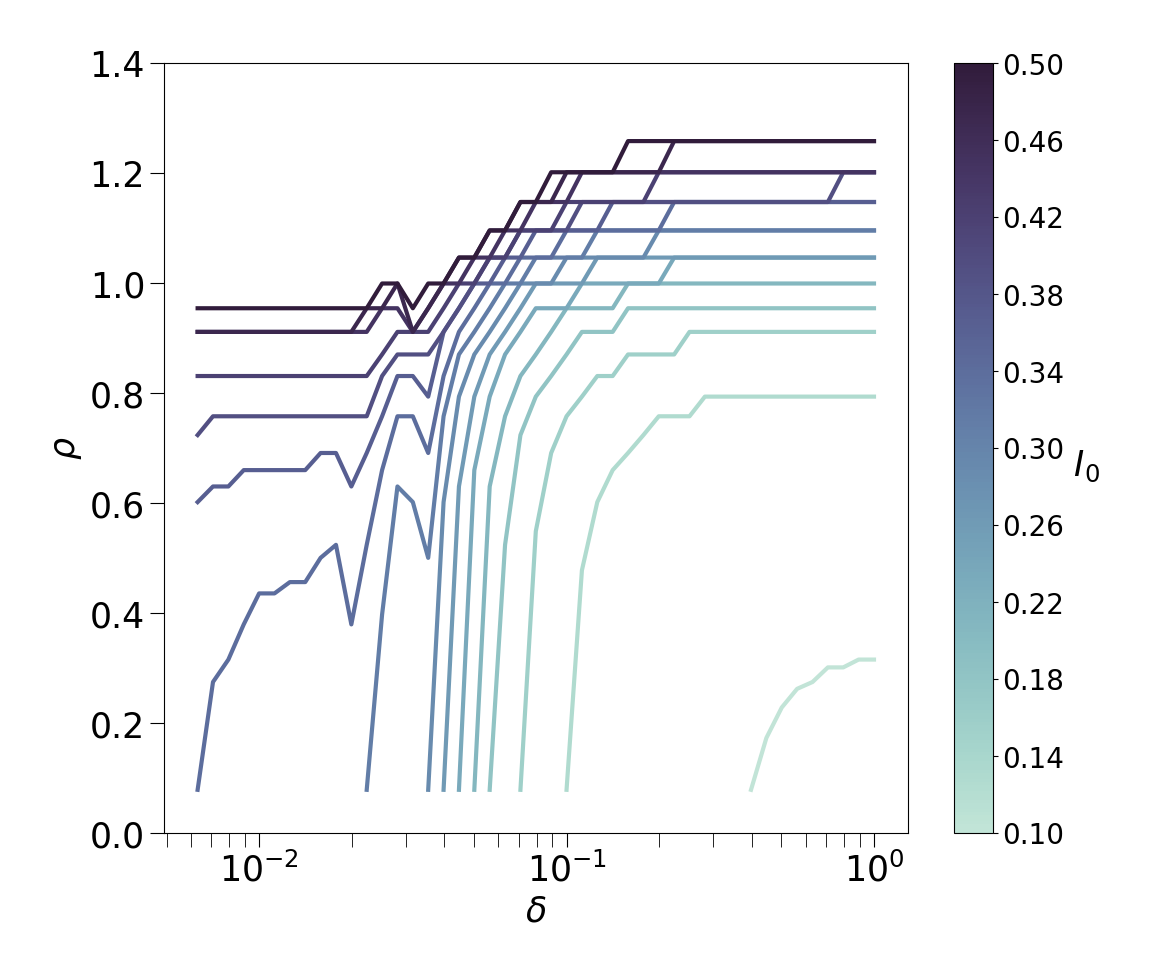}}\\
\vspace{-2mm}
\caption{Experiment 2, Incorporating the DC term with a ramp. (a) and (c) Regions of action potential generation and absence for a given combination of $\rho$ and $\delta$, and for (a) $I_0 = 0.2$ and (c) $I_0=0.4$. In (c), we identified an additional region where persistent excitation was observed. (b) and (d) Particular instances of the solutions for (b) $I_0 = 0.2$ and (d) $I_0 = 0.4$, as indicated in (a) and (c), respectively. (e) Boundary curve separating the regions where no activation was observed and the region where either a single action potential or persistent activation was observed for different values of DC input.
Note that measurements started at $t=100$ in all cases to avoid transient effects due to the HFBS.  Additionally, in Subfigures (a), (c), and (e), the $\delta$ axis is logarithmically spaced.
}\label{fig:FHN_DC}
    
\end{figure}

\appendix
\section{Appendix: Technical proofs.}\label{section_appendix}

\subsection{Derivation of the partially averaged system for steering problem 1}\label{appendix_derivation_pas_slope_HF}
Consider system \eqref{FHN_hf_intro} with input current of the form  \eqref{input_current_slope_general1} with  $\omega \gg 1$. For $0 < t < \frac{1}{\lambda}$, we consider the system 
\begin{equation*}
\left\{
\begin{array}{rl}
\dot{v}&= v -v^3/3 - w +\lambda t \rho \omega \cos(\omega t),\\
\dot{w}&= \varepsilon(v- \gamma w + \beta),\\
v(0) &= v_0,\quad w(0) = w_0.
\end{array}\right.
\end{equation*}
To obtain the averaged system, we use the change of variables $\tilde{v}= v - \lambda t \rho \sin(\omega t), ~\tilde{w} = w$. Next, substituting in the equations for $v$ and $w$, we can write $\dot{\tilde{v}} = f_1(\tilde{v},\tilde{w},t)$ with
\begin{align*}
f_1(\tilde{v},\tilde{w},t) &= (\tilde{v} +  \lambda t \rho \sin(\omega t)) - (\tilde{v} +  \lambda t \rho \sin(\omega t) )^3/3 -\tilde{w},\\
&= (1 - \lambda^2 t^2 \rho^2 /2 )\tilde{v}-\tilde{v}^3/3  -\tilde{w} + \tilde{v} X_1(t) + \tilde{v}^2 X_2(t) + X_3(t),
\end{align*}
where
\begin{equation*}
    X_1(t) = \frac{\lambda^2 t^2 \rho^2}{2} \cos(2\omega t ) ,\quad X_2(t)= -\lambda t \rho \sin(\omega t) ,
\end{equation*}    
\begin{equation*}
    X_3(t) =  \lambda t \rho \sin(\omega t)- (\lambda t \rho \sin(\omega t))^3/3.
\end{equation*}
Analogously, for the second equation, we can write $\dot{\tilde{w}} = f_2(\tilde{v},\tilde{w},t)$ with
\begin{align*}
    f_2(\tilde{v},\tilde{w},t)&=\varepsilon(\tilde{v}-\gamma \tilde{w} + \beta)  + X_4(t),
\end{align*}
where 
\begin{equation*}
    X_4(t) = \varepsilon\lambda t \rho \sin(\omega t).
\end{equation*}
Now, consider the averaging of the right-hand side with respect to the time variable over the interval $[t-\pi/\omega,t+\pi/\omega]$. It is easy to see that $\frac{\omega}{2\pi}\int_{t-\pi/\omega}^{t+\pi/\omega} X_i(s)ds = O(1/\omega)$, which implies 
\begin{equation*}
\hat{f}_1(\tilde{v},\tilde{w},t) =\frac{\omega}{2\pi}\int_{t-\pi/\omega}^{t+\pi/\omega}f_1(\tilde{v},\tilde{w},\tau)d\tau =(1-\lambda^2 t^2 \rho^2/2)\tilde{v} - \tilde{v}^3/3 -\tilde{w} +O(1/\omega),
\end{equation*}
and
\begin{equation*}
\hat{f}_2(\tilde{v},\tilde{w},t) = \frac{\omega}{2\pi}\int_{t-\pi/\omega}^{t+\pi/\omega} f_2(\tilde{v},\tilde{w},\tau)d\tau = \varepsilon (\tilde{v}-\gamma \tilde{w}+\beta)  +O(1/\omega).
\end{equation*}
Using those averaged right-hand side $\hat{f}_1$ and $\hat{f}_2$, and dropping the terms $O(1/\omega)$, we obtain system \eqref{PAS_slope_freq}. For times $t > 1/\lambda$, the derivation is similar, substituting the source $\lambda t \rho \omega \cos(\omega t)$ by $\rho \omega \cos(\omega t)$. In this case, the terms $O(1/\omega)$ are identically zero.

\subsection{Derivation of the partially averaged system for steering problem 2}\label{appendix_derivation_pas_slope_DC}
Consider system \eqref{FHN_hf_intro} with input current of the form  \eqref{input_current_slope_general2} with  $\omega \gg 1$. For $0 < t < \frac{1}{\delta}$, we consider the system 
\begin{equation*}
\left\{
\begin{array}{rl}
\dot{v}&= v -v^3/3 - w +\rho \omega \cos(\omega t)+ \delta t I_0,\\
\dot{w}&= \varepsilon(v- \gamma w + \beta),\\
v(0) &= v_1,\quad w(0) = w_1.
\end{array}\right.
\end{equation*}
Using the change of variables  $\tilde{v} = v - \rho \sin(\omega t)$, $\tilde{w} = w$ we obtain 
\begin{equation*}
\left\{
\begin{array}{rll}
\dot{\tilde{v}}&= \tilde{v} +\rho \sin(\omega t) -(\tilde{v}+\rho \sin(\omega t))^3/3 - \tilde{w} + \delta t I_0 &=: f_1(\tilde{v},\tilde{w},t),\\
\dot{\tilde{w}}&= \varepsilon(\tilde{v}- \gamma \tilde{w} + \beta) + \varepsilon \rho \sin(\omega t) &=: f_2(\tilde{v},\tilde{w},t) .
\end{array}\right.
\end{equation*}
Let us consider the averaging of the right-hand side on $[t-\pi/\omega,t+\pi/\omega]$
\begin{equation*}
\hat{f}_1(\tilde{v},\tilde{w},t) =\frac{\omega}{2\pi}\int_{t-\pi/\omega}^{t+\pi/\omega}f_1(\tilde{v},\tilde{w},\tau)d\tau =(1-\rho^2/2)\tilde{v} - \tilde{v}^3/3 -\tilde{w} + \delta t I_0  +O(1/\omega),
\end{equation*}
and
\begin{equation*}
\hat{f}_2(\tilde{v},\tilde{w},t) = \frac{\omega}{2\pi}\int_{t-\pi/\omega}^{t+\pi/\omega} f_2(\tilde{v},\tilde{w},\tau)d\tau = \varepsilon (\tilde{v}-\gamma \tilde{w}+\beta)  +O(1/\omega).
\end{equation*}
Using those averaged right-hand side $\hat{f}_1$ and $\hat{f}_2$, we obtain the partially averaged system \eqref{FHN_pas}. For time $t > 1/\delta$, the derivation is similar by substituting the source $\rho \omega \cos(\omega t) + \delta t I_0$ by $\rho \omega \cos(\omega t) + I_0$. In this case, the terms $O(1/\omega)$ are identically zero.

\subsection{Proof of Proposition \ref{proposition_parameters}}\label{app_subsection_choice_parameters}
\noindent

Given $\alpha$, $\sigma \in [0,1]$ we look for conditions such that the system 
\begin{equation}\label{parameter_condition_joint}
\begin{cases}
\dot{v} &= (1-\alpha^2 \rho^2/2) v - v^3/3 - w + \sigma I_0,  \\
\dot{w} &= \varepsilon\left( v-\gamma w + \beta \right),
\end{cases}
\end{equation}
has a unique stable equilibrium $(v_0,w_0)$. The coordinate $v$ of the equilibrium can be obtained as the solution to the equation 
\begin{equation}\label{cubic_problem1}
f(v) = v^3 - 3(1-\alpha^2\rho^2/2-1/\gamma)v + 3\beta/\gamma - 3\sigma I_0 = 0.
\end{equation}
Since \eqref{cubic_problem1} is a depressed cubic, the condition for the uniqueness of the real root can be written as
\begin{equation*}
-\left(1-\alpha^2 \rho^2/2-\frac{1}{\gamma}\right)^3 +\frac{1}{4} \left( 3  \frac{\beta}{\gamma}  -3\sigma I_0\right) ^2> 0.
\end{equation*}
This gives us hypothesis \ref{condition_parameters}. For the stability condition,  we linearize \eqref{parameter_condition_joint} around an equilibrium $(v_0, w_0)$ to get the linear system defined by the matrix
\begin{equation*}
L(v_0,w_0) = \left(\begin{array}{cc}
1- \alpha^2 \rho^2/2 - v_0^2 & -1\\
\varepsilon &  - \varepsilon \gamma
\end{array}\right).
\end{equation*}
Therefore, the conditions for stability become
\begin{equation*}
\tr(L(v_0,w_0)) = 1- v_0^2 - \alpha^2 \rho^2/2 - \varepsilon \gamma <0, \,\, \det(L(v_0,w_0)) = \varepsilon \gamma ( v_0^2 - 1 + \alpha^2 \rho^2/2 + \frac{1}{\gamma} ) >0,
\end{equation*}
which can be combined in a more compact expression $v_0^2 > 1-\alpha^2 \rho^2/2- \min\{\varepsilon\gamma, 1/\gamma\}$. At this point, we can see that the requirements \eqref{stronger_stability_1} and \eqref{stronger_stability_2} are indeed stronger stability conditions since they imply the stability of the system. For the stronger condition to be satisfied, we will use that, imposing the strict inequality
\[
v_0^2 >  M = 1- \alpha^2 \rho^2/2,
\]
is enough, since the parameters $\alpha$, $\sigma \in [0,1]$ take values in a compact set. This condition is automatically satisfied when $M\leq 0$. Otherwise, we have to make this condition more explicit. Using \eqref{cubic_problem1}, since $f(v) \to -\infty$ as $v \to -\infty$, and $f(v)\to \infty$ as  $v\to \infty$, and  since $I_0 < \beta/\gamma$ the intermediate value theorem tell us that $v_0 <0$, and therefore the condition
\[
v_0 < -\sqrt{M} = -\sqrt{1-\alpha^2 \rho^2/2},
\]
is equivalent to $f(-\sqrt{M}) > 0$, which gives us condition \ref{condition_parameters1}
\begin{equation*}
M = 1-\alpha^2 \rho^2/2 \leq 0,  \text{ or }\quad 
M^{3/2} + 3 \beta/\gamma -3 \sigma I_0 >3(1-\alpha^2 \rho^2/2 - 1/\gamma)M^{1/2}.
\end{equation*}

Lastly, we note that hypotheses \ref{condition_parameters} and \ref{condition_parameters1} depend on $\alpha$ and $\sigma$. They can be specialized to our cases of interest by noticing that Condition \ref{conditionC1} in Definition \ref{definition_condition_parameters} is equivalent to \ref{condition_parameters} and \ref{condition_parameters1} for $\alpha\in [0,1], \sigma = 0$. Condition \ref{conditionC2} in Definition \ref{definition_condition_parameters} is equivalent to \ref{condition_parameters} and \ref{condition_parameters1} for $\alpha=1, \sigma \in[0,1]$. \qed

\subsection{Proof of Proposition \ref{proposition_easy_conditions}}\label{appendix_proposition_easy_conditions}

Fix values of $\varepsilon>0$, $I_0$ and  $\rho$. We will show that there exist values of $\gamma_0$, $h_0$ such that if condition \eqref{easy_condition} is satisfied, then hypotheses \ref{condition_parameters}, \ref{condition_parameters1} in Definition \ref{defi_hypotheses} are satisfied. This is done by considering stronger but easier-to-verify conditions.

For hypothesis \ref{condition_parameters}, we observe that we can verify instead the stronger condition
\[
3 \frac{\beta}{\gamma }-3 \sigma I_0 \geq 0,\text{ and }
\frac{1}{4} \left( 3  \frac{\beta}{\gamma}  -3 \sigma I_0\right) ^2> \left|1-\alpha^2 \rho^2/2-\frac{1}{\gamma}\right|^3,
\]
which leads us to a condition in $\beta / \gamma$
\[
\frac{\beta}{\gamma} > \sigma  I_0  + \frac{2}{3} \left|1-\alpha^2\rho^2/2-\frac{1}{\gamma}\right|^{3/2}. 
\]
To simplify the computations, we can further assume that $\frac{1}{\gamma} \leq 1$, which allows us to write an even stronger condition, but it is independent of $\alpha$, $\sigma\in [0,1]$, with a right-hand side which is independent on $\gamma$
\[
\frac{\beta}{\gamma} >  \max\{I_0,0\}  + \frac{2}{3} \left|1+\rho^2/2+1\right|^{3/2} =: h_1. 
\]
This tells us that given any values of $\varepsilon>0$, $I_0\in \R$ and  $\rho\in \R$, whenever $\beta/\gamma > h_1$ and $\gamma \geq 1$ then the hypothesis \ref{condition_parameters} is satisfied. 

For hypothesis \ref{condition_parameters1}, we first notice that if $1-\rho^2/2 \leq0$, we are done. Otherwise, we can assume that  $M(\alpha) = 1- \alpha^2 \rho^2/2 >0$ for some values of $\alpha$. If that is the case, let $\alpha\in[0,1]$ such that $M(\alpha)>0$, we verify the stronger condition 
\[
3 \frac{\beta}{\gamma } - 3 \sigma I_0 > 3 |1- \alpha^2 \rho^2/2 - 1/\gamma | M ^{1/2}.
\]
Since we want a condition that is independent of $\alpha$, $\sigma \in [0,1]$, we use that $0<M(\alpha) \leq 1$ and $|M(\alpha) - 1/\gamma| \leq M(\alpha) + \frac{1}{\gamma} \leq 2$, because $\gamma \geq 1$, which lead us to the following stronger condition 
\[
\frac{\beta}{\gamma } > \max\{I_0,0\}   + 2 =:h_2,
\]
which tells us that, whenever $\beta/\gamma > h_2$ and $\gamma \geq 1$ then hypothesis \ref{condition_parameters1} is satisfied. So far, we have proven that whenever
\[
\beta/\gamma > \max\{h_1,h_2\} =: h_0, \text{ and }\gamma \geq 1,
\]
both \ref{condition_parameters} and \ref{condition_parameters1} are satisfied. For the second part of the proof, we need the following Lemma
\begin{lemma}\label{lemma_estimate_parameters_appendix}
Suppose that $\rho^2 < \frac{4}{3}$, $\frac{\beta}{\gamma}\geq \max\{I_0,0\}$. Then given $\eta>0$, there exists $\tilde{\gamma}>0$ such that for any $\gamma\geq \tilde{\gamma}$ we have 
\begin{equation}\label{additional_condition_parameters}
\left(1+\frac{1}{\eta\gamma}\right)^{1/2}\left( 2-\frac{3}{2}\rho^2 - \frac{3+1/\eta}{\gamma}\right) + 3\frac{\beta}{\gamma} - 3 \sigma I_0>0.
\end{equation}
\end{lemma}
\begin{proof}
First notice that $A = 2 - \frac{3}{2}\rho^2 =\frac{3}{2} (\frac{4}{3} - \rho^2) > 0$. Next given $\eta > 0$, we choose $\tilde{\gamma}>0$ large enough so that
\[
A = 2 - \frac{3}{2}\rho^2 > \frac{1}{\tilde{\gamma}}\left(3+\frac{1}{\eta}\right),
\]
with that choice, we can guarantee that for all  $\gamma \geq \tilde{\gamma}$ and $\alpha\in[0,1]$ we have
\[
\left( 2-\frac{3}{2}\alpha^2\rho^2 - \frac{1}{\gamma}\left(3+\frac{1}{\eta}\right)\right) >0,
\]
and because $\frac{\beta}{\gamma} \geq \max\{I_0,0\}$ this implies \eqref{additional_condition_parameters} 
for all $\sigma\in[0,1]$.
\end{proof}
Next, we will borrow an argument from \cite{cerpaApproximationStabilityResults2023a}, to show that under the same conditions as Lemma \ref{lemma_estimate_parameters_appendix}, given $\eta>0$ there exists $\tilde{\gamma}>0$ such that for all $\gamma \geq \tilde{\gamma}$ the solution $(X_1, Y_1)$ of \eqref{PAS_slope_freq_instantaneous_alpha} satisfy
\begin{equation}\label{estimate_dependence_gamma}
0< \frac{1}{\min_{\alpha\in[0,1]}\left(X_1(\alpha)^2-1+ \rho^2 \alpha^2/2\right)}\leq \eta \gamma.
\end{equation}
As an intermediate step, we first show the following
\begin{equation}\label{bounds_equilibrium}
- \max\{\sqrt{3},\beta\} \leq X_1(\alpha) < -\sqrt{1+ \frac{1}{\eta \gamma}}.
\end{equation}
To obtain this bound, we use that the solution of \eqref{PAS_slope_freq_instantaneous_alpha} can be obtained by solving the following cubic equation
\begin{equation}\label{defi_polinomial_h}
f(x) = x^3 - 3(1-\alpha^2\rho^2/2-1/\gamma)x + 3\beta/\gamma - 3\sigma I_0  = 0.
\end{equation}
Because of \ref{condition_parameters} and \ref{condition_parameters1}, Equation \eqref{defi_polinomial_h} has a unique solution. To bound such a solution, call it $x = x_0$; we use that because the intermediate value theorem tells us that $a<x_0<b$ if and only if $f(a)<0<f(b)$. Next, given $\eta >0$, let us compute
\begin{align*}
f\left(-\left(1+\frac{1}{\eta\gamma}\right)^{1/2}\right) &= -\left(1+\frac{1}{\eta\gamma}\right)^{3/2} + 3\left(1-\alpha^2\rho^2/2- \frac{1}{\gamma}\right)\left(1+\frac{1}{\eta\gamma}\right)^{1/2}+ 3\frac{\beta}{\gamma} - 3 \sigma I_0 \notag,\\
&= \left(1+\frac{1}{\eta\gamma}\right)^{1/2}\left( 2-\frac{3}{2}\alpha^2 \rho^2 - \frac{3+1/\eta}{\gamma}\right) + 3\frac{\beta}{\gamma}  - 3 \sigma I_0,
\end{align*}
Because of Lemma \ref{lemma_estimate_parameters_appendix} we know that this quantity is positive for $\rho^2 < \frac{4}{3}$, $\frac{\beta}{\gamma} \geq \max\{I_0,0\}$, $\gamma \geq \tilde{\gamma}$, $\alpha\in[0,1]$, $\sigma \in[0,1]$. Hence $\left(- \sqrt{1+\frac{1}{\eta\gamma}}\right)>X_1(\alpha)$, giving the upper bound in \eqref{bounds_equilibrium}. To establish the lower bound, we evaluate $f$ for some special values, 
\begin{align*}
f(-\sqrt{3})&= -\frac{3}{\gamma} (\sqrt{3} -\beta) - \frac{3\sqrt{3}}{2}\alpha^2\rho^2 - 3 \sigma I_0,\\
f(-\beta)&= -\beta(\beta-\sqrt{3})(\beta+\sqrt{3}) - \frac{3}{2}\beta \alpha^2\rho^2 - 3 \sigma I_0.
\end{align*}
Because $I_0 \geq 0$, we can guarantee that if
$\beta < \sqrt{3}$ then $f(-\sqrt{3})<0$ and if $\beta \geq \sqrt{3}$ we get $f(-\beta)\leq 0$, for all values of $\alpha,\sigma \in[0,1]$. This implies $v_0\geq\min\{-\beta,-\sqrt{3}\}$, which gives us the second part of \eqref{bounds_equilibrium}. The next step is to show that \eqref{bounds_equilibrium} implies \eqref{estimate_dependence_gamma}. Because \eqref{bounds_equilibrium} tells us that $X_1(\alpha) <0$, we can take squares to obtain
\[
\frac{1}{\eta \gamma} + \frac{\alpha^2 \rho^2}{2} \leq X_1(\alpha)^2 -1 + \frac{\alpha^2 \rho^2}{2} \leq
\max\{\sqrt{3},\beta\}^2 - 1 + \frac{\alpha^2 \rho^2}{2},
\]
taking the minimum in $\alpha\in[0,1]$ and dividing we obtain
\[
0 < \frac{1}{\max\{\sqrt{3},\beta\}^2 - 1  } \leq  \frac{1}{\min_{\alpha\in[0,1]} \left(X_1(\alpha)^2 -1 + \frac{\alpha^2 \rho^2}{2}\right)}  \leq \gamma\eta ,
\]
which gives us \eqref{estimate_dependence_gamma}. We are finally in good standing to verify hypotheses  \ref{condition_approximation_HF} and \ref{condition_approximation_DC} are nonempty. This is done by finding appropriate values for $\omega_0>0$ and $\gamma_0\geq 1$ for which their hypotheses are valid. To do this, let $\eta = \frac{1}{8e}$, then applying previous result we can guarantee that for some $\gamma_0\geq 1$, the solution $(X_1, Y_1)$ of \eqref{PAS_slope_freq_instantaneous_alpha} satisfy that if $\rho^2 < \frac{4}{3}$, $\frac{\beta}{\gamma} > \max\{I_0,0,h_0\}$, $\gamma \geq \gamma_0$, $\alpha\in[0,1]$, $\sigma \in[0,1]$ then 
\[
0 < \frac{1}{\min_{\alpha\in[0,1]} \left(X_1(\alpha)^2 -1 + \frac{\alpha^2 \rho^2}{2}\right)} \leq \eta = \frac{1}{8 e}.
\]
Because of the continuity, for some $M>0$ we have that whenever $\sup_{\alpha\in[0,1]}|X_1(\alpha) - h(\alpha)| < M$ then 
\[
0<\frac{1}{\min_{\alpha\in[0,1]} \left(
h(\alpha)^2 -1 + \frac{\alpha^2 \rho^2}{2}\right)} < \frac{1}{4e}.
\]
For the remainder of the proof, we fix the values of the parameters $\beta$, $\gamma$ (the parameters $\varepsilon$, $\rho$, $I_0$ are fixed in the statement of the proposition). Because of hypothesis \ref{condition_parameters} we know from Proposition \ref{prop_slope_freq} that there exits $\hat{\lambda}$, $\hat{\mu}>0$ such that for all $0<\lambda < \hat{\lambda}$ if the initial data satisfy 
\[
\mu = \|(V(0) - v_0, W(0)-w_0 )\|_2 \leq \mu^*,
\]
then for all $t\in [0,1/\lambda]$ get that 
\[
|V(t) - X_1(t)|^2 + |W(t) - Y_1(t)|^2 \leq C \mu^2 + C_1 \lambda,
\]
and therefore by choosing $\tilde{\mu}\leq \hat{\mu}$, $\tilde{\lambda}\leq \hat{\lambda}$ so that  $C \tilde{\mu}^2 + C_1 \tilde{\lambda} \leq M^2$ we can guarantee that for all $0<\lambda <\tilde{\lambda}$ and $\mu \leq \tilde{\mu}$ we have
\[
\sup_{\alpha\in[0,1]}|X_1(\alpha) - V(\alpha)| \leq M,
\]
which implies that 
\begin{equation}\label{estimate_quotient_V}
0<\frac{1}{\min_{s\in[0,1/\lambda]} \left(
V(s)^2 -1 + \frac{s^2 \rho^2}{2}\right)} < \frac{1}{4e}.
\end{equation}
Next, look at the condition in hypothesis \ref{condition_approximation_HF}
\begin{equation*}
0< I = \frac{1}{\gamma}\exp\left(\frac{\rho^2}{2\omega} + \frac{4 |\rho|}{\omega} \max_{s\in[0,\frac{1}{\lambda}]}|V(s)|\right) \left(\min_{s\in[0,\frac{1}{\lambda}]}\left(|V(s)|^2 + (\rho \lambda s )^2/2  -1\right)- \frac{1}{\omega}g(\lambda,\rho)\right)^{-1} <\frac{1}{2},
\end{equation*}
where $g(\lambda,\rho)= \frac{\rho^2 \lambda}{2} + 2|\rho| \lambda\max_{ s\in[0, \frac{1}{\lambda}]}|V(s)| + 2 |\rho| \max_{s\in[0,\frac{1}{\lambda}]}|\dot{V}(s)|$. Because of Proposition \ref{prop_properties_pas_HF} we know that there exists $\lambda^* \leq \tilde{\lambda}$, and $\mu^*\leq \tilde{\mu}$ such that for all $0<\lambda <\lambda^*$ and $\mu\leq\mu^*$ the quantities $\max_{s\in[0, \frac{1}{\lambda}]}|V(s)|$ and $\max_{s\in[0, \frac{1}{\lambda}]}|\dot{V}(s)|$ are bounded independently of $\omega>0$ and
\[
\min_{s\in [0, \frac{1}{\lambda}]}\left(|V(s)|^2 + (\rho \lambda s )^2/2-1\right) >0.
\]
Then we can take $\omega \geq \omega_0(\varepsilon, \beta,\gamma, \rho, I_0)$ large enough to guarantee that 
\begin{itemize}
\item[(i)] $\frac{\rho^2}{2\omega} + \frac{4 |\rho|}{\omega} \max_{s\in[0, 1/\lambda]}|V(s)| < 1$, 
\item[(ii)] $\frac{1}{\omega}g(\lambda,\rho)\leq \frac{1}{2}\left(\min_{s\in [0, \frac{1}{\lambda}]}|V(s)|^2 + (\rho \lambda s )^2/2-1\right)$.
\end{itemize}
With that choice of $\omega$ and estimate \eqref{estimate_quotient_V} we obtain
\[
0< I \leq \frac{2e }{\gamma} \frac{1}{\min_{s\in[0,1/\lambda]}(|V(s)|^2-1 + (\rho\lambda s)^2/2)}  < \frac{1}{2},
\]
which gives us hypothesis \ref{condition_approximation_HF}.
The corresponding conditions to verify Hypothesis \ref{condition_approximation_DC} are obtained analogously. This concludes the proof of Proposition \ref{proposition_easy_conditions}. \qed

\bibliographystyle{habbrv}
\bibliography{bibliography}
\end{document}